\documentclass[10pt,journal, final,twocolumn]{IEEEtran}

\usepackage{amsfonts}
\usepackage[cmex10]{amsmath}
\usepackage{amssymb}
\usepackage{amsthm}
\usepackage{cite}
\usepackage{stmaryrd} 
\usepackage[all,cmtip,line]{xy}
\usepackage[pdftex]{graphicx, color}
\usepackage{dsfont}
\usepackage{array} 
\usepackage{url}
\usepackage[colorlinks=true]{hyperref}

\DeclareMathOperator{\Tr}{Tr}
\DeclareMathOperator{\Range}{Range}

\DeclareMathOperator{\grad}{grad}
\DeclareMathOperator{\Hess}{Hess}
\DeclareMathOperator{\Crit}{Crit}
\DeclareMathOperator{\codim}{codim}

\DeclareMathOperator{\Ad}{Ad}

\DeclareMathOperator{\Aut}{Aut}
\DeclareMathOperator{\End}{End}

\DeclareMathOperator{\im}{im}

\newcommand{\jmdstack}[2]{\genfrac{}{}{0pt}{}{#1}{#2}}
\newcommand{\rmd}{\mathrm{d}}
\newcommand{\rmT}{\mathrm{T}}
\newcommand{\BH}{\mathcal{B}(\mathcal{H})}
\newcommand{\MH}{\mathcal{M}(\mathcal{H})}

\newcommand{\UH}{\mathrm{U}(\mathcal{H})}

\newcommand{\uH}{\mathrm{u}(\mathcal{H})}

\newcommand{\SUH}{\mathrm{SU}(\mathcal{H})}

\newcommand{\PUH}{\mathrm{PU}(\mathcal{H})}
\newcommand{\id}{\mathrm{id}}
\newcommand{\identity}{\mathds{1}}
\newcommand{\ES}{\mathcal{V}}

\newcommand{\GrES}[2]{\mathrm{Gr}_{#1}\big(\ES_{#2}\big)}
\newcommand{\AAdag}{\mathcal{A}^{2}}
\newcommand{\AAdagSq}{\mathcal{A}^{4}}
\newcommand{\Herm}{\mathbb{H}(\mathcal{H})}
\newcommand{\phase}{\phi}

\newtheorem{theorem}{Theorem}
\newtheorem{lemma}{Lemma}
\newtheorem{definition}{Definition}
\newtheorem{example}{Example}

\author{Jason M. Dominy, Tak-San Ho, and Herschel A. Rabitz
\thanks{Manuscript received February 9, 2011; revised MMMMMMM DD, 2013. This work was supported, in part, by U.S. Department of Energy (DOE) Contract No. DE-AC02-76-CHO-3073 through the Program in Plasma Science and Technology at Princeton.  Partial support is also acknowledged from the DOE grant No. DE-FG02-02ER15344.}
\thanks{J. M. Dominy was with the Program in Applied and Computational Mathematics, Princeton University, Princeton NJ 08544, during the writing of this paper.  He is currently with the Center for Quantum Information Science and Technology, University of Southern California, 90089 (e-mail: jdominy@usc.edu)}
\thanks{T.-S. Ho is with the Department of Chemistry, Princeton University, Princeton NJ 08544 (e-mail: tsho@princeton.edu)}
\thanks{H. A. Rabitz is with the Department of Chemistry and the Program in Applied and Computational Mathematics, Princeton University, Princeton NJ 08544 (e-mail: hrabitz@princeton.edu)}
\thanks{Document Object Identifier }
}

\title{Characterization of the Critical Sets of Quantum Unitary Control Landscapes}

\begin{document}

\maketitle

\begin{abstract}This work considers various families of quantum control landscapes (i.e. objective functions for optimal control) for obtaining target unitary transformations as the general solution of the controlled Schr\"odinger equation.  We examine the critical point structure of the kinematic landscapes $J_{F}(U) = \|(U-W)A\|^{2}$ and $J_{P}(U) = \|A\|^{4} - |\Tr( AA^{\dag} W^{\dag}U)|^{2}$ defined on the unitary group $\UH$ of a finite-dimensional Hilbert space $\mathcal{H}$.  The parameter operator $A\in\BH$ is allowed to be completely arbitrary, yielding an objective function that measures the difference in the actions of $U$ and the target $W$ on a subspace of state space, namely the column space of $A$. The analysis of this function includes a description of the structure of the critical sets of these kinematic landscapes and characterization of the critical points as maxima, minima, and saddles.  In addition, we consider the question of whether these landscapes are Morse-Bott functions on $\UH$.  Landscapes based on the intrinsic (geodesic) distance on $\UH$ and the projective unitary group $\PUH$ are also considered.  These results are then used to deduce properties of the critical set of the corresponding dynamical landscapes.
\end{abstract}


\section{Introduction}
\IEEEPARstart{A}{n} important application of quantum optimal control theory is the generation of target quantum logic gates for quantum information processing.  The goal of such optimal control is to arrange the dynamics such that the desired logical gate is realized as the final time unitary evolution operator, which is the general solution of the controlled Schr\"odinger equation.  In most applications, the optimization goal is not a single unitary operator, but a family of logically equivalent operators.  For example, since the global phase is not observable, the goal may be any unitary operator that is equivalent to the target gate up to global phase.  Likewise, in some cases only a subspace of the Hilbert space $\mathcal{H}$ of states may be used for the quantum register, so that any unitary propagator should be acceptable that acts as the target gate on that subspace.  In contrast to other quantum control problems, for example the maximization of a quantum mechanical observable, there is no unique or natural choice for the objective function against which the optimization is performed.  Indeed, \emph{any} smooth function $J:\UH\to\mathbb{R}$ with a global minimum at the target unitary gate or gates is a candidate objective for the unitary problem.  But, as we will see, some choices may exhibit more favorable convergence and other properties.

The theory of quantum control landscapes has been developed over a series of papers, including \cite{Rabitz2004, Rabitz2005, Rabitz2006, Rabitz2006a, Wu2008,Wu2008a,Hsieh2008, Hsieh2008a,Ho2009}, as a way to think about the problem of finding optimal solutions within quantum control.  This is pursued principally by building up a picture of the topography of the objective function as a landscape over the space of all admissible controls, typically through analysis of the structure of the set of critical points of the objective function.  This provides direct information about the gradient flow associated with the landscape.  For example, the presence of a local maximum or minimum can act as a ``trap'' for the gradient flow or its time-reversal, respectively.  And although saddles do not trap the flow, the flow can be greatly slowed in close proximity to a saddle.  While, for a given objective function, gradient ascent/descent may not be the most efficient method for finding optimal controls, the topography of the landscape and its impact on the behavior of the gradient flow offers insights into the expected performance of classes of algorithms (local deterministic algorithms versus more non-local stochastic algorithms, for example).  As a consequence, a quantum control landscape analysis will typically begin with the identification of the set of critical points.

The critical points of the kinematic landscape having been identified, they may then be characterized as local maxima, local minima, and saddles.  As has been demonstrated for other classes of kinematic quantum control landscapes, the landscapes considered in this work will turn out to have global maxima and minima, but no other local extrema capable of impeding optimization.  Moreover, the critical sets will be shown generally to comprise disjoint submanifolds, and these submanifolds are nondegenerate in the Morse-Bott sense \cite{Bott1954, Atiyah1983, Nicolaescu2007}.  In other words, the null space of the Hessian of $J$ and the tangent space of the critical submanifold coincide at each critical point $U\in\UH$.  This condition identifies the kinematic landscape as a Morse-Bott function, which is interesting for at least two reasons.  First, certain results about the convergence of the gradient flow may be proved for Morse-Bott functions, in particular that (on a compact manifold) the gradient flow always converges to a critical point \cite{Helmke1996}.  Second, the identification of the null space of the Hessian and the tangent space of the critical submanifold is important for certain numerical methods, such as second order D-MORPH \cite{Beltrani2011}, that are designed to explore the critical sets.

Let $\mathbb{K}$ denote the space of admissible control functions.  For the present analysis, $\mathbb{K}$ will be $L^{2}([0,T];\mathbb{R})$, the space of square integrable real-valued functions on the interval $[0,T]$, where $T$ is some fixed final time over which the controlled dynamics take place.  Let $\mathcal{H}$ be a complex Hilbert space of dimension $N<\infty$ and let $\BH$ denote the space of bounded linear operators on $\mathcal{H}$ endowed with the real Hilbert-Schmidt inner product $\langle A, B\rangle = \Re\Tr(A^{\dag}B)$.  $\UH\subset\BH$ will denote the unitary group on $\mathcal{H}$ endowed with the Riemannian metric induced by the Hilbert-Schmidt inner product, and $\uH$ will denote the corresponding Lie algebra of skew-Hermitian operators acting on $\mathcal{H}$.  Also let $U_{T}:\mathbb{K}\rightarrow \UH$ denote the map, defined implicitly by the Schr\"odinger equation in the dipole approximation 
\IEEEpubidadjcol
\begin{align}
	i\hbar\frac{\rmd U}{\rmd t}(t,t_{0}) & = \big(H_{0} - \mu\mathcal{E}(t)\big)U(t,t_{0}), & U(t_{0},t_{0}) = \identity,
\end{align}
such that $U_{T}(\mathcal{E}) = U(T,0)[\mathcal{E}]$ is the unitary propagator at time $T$ for the control field $\mathcal{E}\in\mathbb{K}$. Finally, for any candidate objective function $J:\UH\to\mathbb{R}$ (the ``kinematic landscape''), let $\tilde{J}:\mathbb{K}\rightarrow \mathbb{R}$ (the ``dynamical landscape'') be the composition $\tilde{J} = J\circ U_{T}$.  Then 
\begin{equation}
	\grad \tilde{J}(\mathcal{E}) = (\rmd_{\mathcal{E}}U_{T})^{*}\Big(\grad J\big(U_{T}(\mathcal{E})\big)\Big),
\end{equation}
where $(\rmd_{\mathcal{E}}U_{T})^{*}$ is the operator adjoint of the differential $\rmd_{\mathcal{E}}U_{T}$.  Much of the important information about the nature of the gradient flow of $\tilde{J}$ is embodied in the critical points of this landscape, i.e. those fields $\mathcal{E}\in\mathbb{K}$ for which $\grad \tilde{J}(\mathcal{E}) = 0$.  Any $\mathcal{E}\in\mathbb{K}$ such that $\rmd_{\mathcal{E}}U_{T}$ is full rank and $\grad J\big(U_{T}(\mathcal{E})\big) = 0$ (so-called ``regular'' critical points) will satisfy the condition.  There may be other critical points where $\rmd_{\mathcal{E}}U_{T}$ is rank-deficient and $\grad J\big(U_{T}(\mathcal{E})\big)$ may or may not be zero (``singular'' points).  Consideration of such singular points is important for a complete understanding on the dynamical control landscape.  However, since $U_{T}$ is a highly nonlinear map from an infinite-dimensional space to a finite-dimensional space, singular points are expected to be rare and will not be considered in the present analysis.  Singular points and their role in quantum control landscapes have only recently begun to be studied \cite{Wu2012,  Fouquieres, Pechen2011, Rabitz2012}.

Several classes of landscapes for generating target unitary transformations will be considered.  They include $J_{F}(U) = \|(U-W)A\|^{2}$ and $J_{P}(U) = \|A\|^{4} - |\Tr( AA^{\dag} W^{\dag}U)|^{2}$ for some fixed target $W\in\UH$ and some arbitrary fixed $A\in\BH $, as well as the corresponding landscapes using the intrinsic (geodesic) distance on the unitary group $\UH$ and the projective unitary group $\PUH$, rather than the norm (Euclidean) distance as in $J_{F}$ and $J_{P}$.  The parameter operator $A$, though arbitrary in our analysis, can have some implications for the landscape topography and should be chosen carefully in application.  $A$ might be chosen to be a partial isometry, for example, so that $AA^{\dagger}$ is a projection onto a subspace of $\mathcal{H}$.  Such a form may be desirable in quantum information when only a subspace of $\mathcal{H}$ is designated as the quantum register.  Alternatively, $A$ might be chosen to be nondegenerate, which can have the effect of simplifying the landscape topography by making all critical submanifolds zero-dimensional, and perhaps making optimization easier.  

Landscapes of the form $J_{F}$ and $J_{P}$ have been studied in the past \cite{Rabitz2005, Hsieh2008a, Ho2009}.  The present paper extends these various works by broadening the families of landscapes under consideration, describing the structures of the critical submanifolds and Hessian eigenbundles (vector bundles formed from the Hessian eigenspaces along critical submanifolds; see Figure \ref{fig:vecBund}), and directly addressing the issue of Morse-Bott nondegeneracy of the critical submanifolds.  
\begin{figure}
	\centering
	\includegraphics[scale=0.5,clip = true]{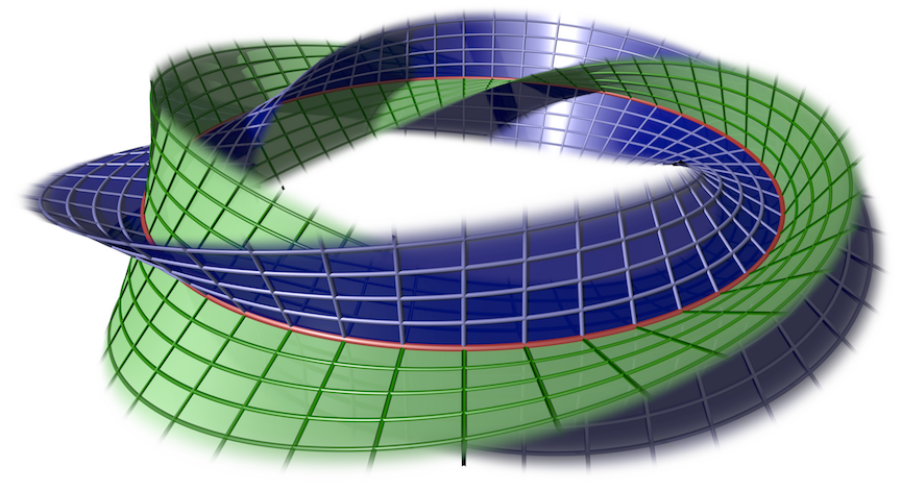}
	\caption{Example of two orthogonal rank 1 vector bundles (i.e. line bundles) over $S^{1}$ in $\mathbb{R}^{3}$.  Analogously, each Hessian eigenbundle considered in this paper is a rank $m$  vector bundle over a critical submanifold of $\UH$, where $m$ is the multiplicity of the corresponding eigenvalue.  It is formed by ``gluing'' together the $m$-dimensional $\lambda$-eigenspaces at each point of the critical submanifold for a given Hessian eigenvalue $\lambda$.}
	\label{fig:vecBund}
\end{figure}

The paper is organized as follows.  Sections \ref{sec:frobkinematic} and \ref{sec:phaseInvKinematic} describe the critical points of the kinematic landscapes $J_{F}$ and $J_{P}$.  The two additional landscapes based on geodesic distance are presented and analyzed in Section \ref{sec:intrinsicDist}.  In Section \ref{sec:dynamicalAnalysis}, these results are related back to the dynamical landscapes.  The overall results are summarized in Section \ref{sec:summary}.  Two appendices are included which provide a proof of the infinite Fr\'echet differentiability of the control-to-propagator map $U_{T}$ and a derivation of the gradient of the geodesic distance landscape on $\UH$.

\bigskip

\section{Kinematic Critical Point Analysis of Certain Phase-Dependent Landscapes\label{sec:frobkinematic}}
For now we put aside the dynamical component of the map and focus just on the critical point analysis of the kinematic map 
\begin{equation}
	J_{F}(U) = \|(U-W)A\|^{2} = 2\|A\|^{2} - 2\Re\Tr( \AAdag W^{\dag}U),
\end{equation} 
defined as a function on the unitary group: $J_{F}:\UH\to\mathbb{R}$, where $\AAdag:=AA^{\dag}$ is Hermitian and positive semi-definite.  In contrast to the landscape $J_{P}$ that will be the subject of the next section, the value of $J_{F}(U)$ depends upon the global phase of $U$.  However, if $ \AAdag $ has a null space, then $J_{F}(U)$ is invariant to the action of $U$ on that null space.  Generally speaking, the parameter operator $A\in\BH$ may be chosen to vary the relative weights of different parts of the unitary operator.  This freedom allows for the possibility that $A$ is a projection operator as discussed in \cite{Palao2003}.

\subsection{Critical Point Identification}
The operator $ \AAdag $ induces a natural orthogonal decomposition of $\mathcal{H}$ into the eigenspaces of $ \AAdag $, $\mathcal{H} = \bigoplus_{i=0}^{\kappa}\ES_{i}$, under which $ \AAdag = \bigoplus_{i=0}^{\kappa}\omega_{i}^{2}\identity_{\ES_{i}}$ where $0 = \omega_{0}^{2}<\omega_{1}^{2}<\dots <\omega_{\kappa}^{2}$ are the eigenvalues of $ \AAdag $.  Note that $\omega_{0}^{2}=0$ is a special case, and that $\ES_{0}$, the null space of $ \AAdag $, may be a trivial (zero-dimensional) subspace, while the other $\ES_{i}$ are assumed to be nontrivial.  This is done because of the special significance of $\ES_{0}$ in the analysis that follows.  The dimension of subspace $\ES_{i}$ (i.e. the multiplicity of eigenvalue $\omega_{i}^{2}$) will be denoted by $n_{i}$. 

The remainder of Section \ref{sec:frobkinematic} will be concerned with proving the results encapsulated in the following theorem and illustrating them with examples.
\begin{theorem}
	The kinematic objective $J_{F}$ is a Morse-Bott function on $\UH$ exhibiting $M = \prod_{i=1}^{\kappa}(n_{i}+1)$ connected critical submanifolds, each of which is isomorphic to a product of Grassmann manifolds and a unitary group:
\begin{subequations}
\begin{align}
	&\Crit(J_{F})\nonumber\\
	& = \bigsqcup_{\jmdstack{0\leq \nu_{i}\leq n_{i}}{\forall i=1,\dots,\kappa}} \Big\{U = W\bigoplus_{i=0}^{\kappa} \mathcal{X}_{i}\;:\; \mathcal{X}_{i}\in\mathrm{U}(\ES_{i}) \text{ for } 0\leq i\leq \kappa,\nonumber\\
	& \qquad \qquad \qquad \text{ and }\mathcal{X}_{i}^{2} = \identity_{\ES_{i}} \text{ and }\nonumber\\
	& \qquad \qquad \qquad \Tr(\mathcal{X}_{i}) = n_{i}-2\nu_{i} \text{ for } 1\leq i\leq \kappa\Big\}\\
	& \simeq \bigsqcup_{\jmdstack{0\leq \nu_{i}\leq n_{i}}{\forall i=1,\dots,\kappa}}\mathrm{U}(\ES_{0})\oplus\GrES{\nu_{1}}{1}\oplus\cdots\oplus\GrES{\nu_{\kappa}}{\kappa},
\end{align}
\end{subequations}
where $\GrES{\nu_{i}}{i}$ denotes the Grassmann manifold of all $\nu_{i}$-dimensional linear subspaces of $\ES_{i}$.  The critical submanifold described by a particular set of indices $(\nu_{1},\dots,\nu_{\kappa})$ corresponds to a critical value $J_{F} = 4\sum_{i=1}^{\kappa}\omega_{i}^{2}\nu_{i}$ and has dimension 
\begin{equation}
	\mathcal{N}_{0} = n_{0}^{2} + 2\sum_{i=1}^{\kappa}\nu_{i}(n_{i}-\nu_{i}),
\end{equation}
while the ranks of the negative and positive Hessian eigenbundles (i.e. the numbers of negative and positive Hessian eigenvalues) on this submanifold are 
\begin{align}
	\mathcal{N}_{-} & = \sum_{i=1}^{\kappa}\nu_{i}^{2} + 2\sum_{0\leq i<j}^{\kappa}n_{i}\nu_{j}\\
	\mathcal{N}_{+} & = \sum_{i=1}^{\kappa}(n_{i}-\nu_{i})^{2} + 2\sum_{0\leq i<j}^{\kappa}n_{i}(n_{j}-\nu_{j}).
\end{align}
Of these critical submanifolds, exactly one (corresponding to $\nu_{i} = n_{i}$ for all $i=1,\dots, \kappa$) is the set of global maxima and one (corresponding to $\nu_{i} = 0$ for all $i=1,\dots, \kappa$) is the set of global minima, both being isomorphic to $\mathrm{U}(n_{0})$.  The remaining $M-2$ critical submanifolds are all saddles, so that $J_{F}$ admits no local traps.
\label{thm:FrobeniusSummary}
\end{theorem}

The differential of $J_{F}$, $\rmd_{U}J_{F}:\rmT_{U}\UH\rightarrow \mathbb{R}$, is given by 
\begin{align}
	\rmd_{U}J_{F}(\delta U) & = -2\Re\Tr( \AAdag W^{\dag}\delta U) \nonumber\\
	& = \langle U \AAdag W^{\dag}U-W \AAdag , \delta U\rangle,
\end{align}
where the last step includes a projection of $-2W \AAdag $ into the tangent space $\rmT_{U}\UH$. Therefore, 
\begin{equation}
	\grad J_{F}(U)  = U \AAdag W^{\dag}U-W \AAdag
\end{equation}
and a critical point of $J_{F}$ is a $U$ such that $U \AAdag W^{\dag}U = W \AAdag $.  Let $\mathcal{X} = W^{\dag}U\in\UH$.  Then the critical point condition becomes $\mathcal{X}^{\dag} \AAdag = \AAdag \mathcal{X}$.  This result also implies that $\mathcal{X} \AAdagSq =  \AAdagSq \mathcal{X}$, since $\mathcal{X}^{\dag} \AAdagSq \mathcal{X} = (\mathcal{X}^{\dag} \AAdag )( \AAdag \mathcal{X}) = ( \AAdag \mathcal{X})(\mathcal{X}^{\dag} \AAdag ) = \AAdagSq $.  Because the eigendecomposition of $\mathcal{H}$ induced by $ \AAdagSq $ is the same as for $ \AAdag $, $\mathcal{H} =\bigoplus_{i=0}^{\kappa}\ES_{i}$, the fact that $\mathcal{X}$ commutes with $ \AAdagSq $ implies that $\mathcal{X} = \bigoplus_{i=1}^{\kappa}\mathcal{X}_{i}$ with $\mathcal{X}_{i}\in\mathrm{U}(\ES_{i})$ for each $i$.  Then $\mathcal{X}^{\dag} \AAdag = \AAdag \mathcal{X}$ implies that $\omega_{i}^{2}\mathcal{X}_{i}^{\dag} = \omega_{i}^{2}\mathcal{X}_{i}$ for each $i=0,\dots,\kappa$, so that for $i=1,\dots,\kappa$, $\mathcal{X}_{i} = \mathcal{X}_{i}^{\dag}$ is both unitary and Hermitian, and therefore has eigenvalues $\pm 1$.  In other words, each $\mathcal{X}_{i}$ is a unitary involution: $\mathcal{X}_{i}^{2} = \identity_{\ES_{i}}$.  Consequently, for $i=1,\dots,\kappa$, there exists a further orthogonal decomposition of $\ES_{i}$ into the positive and negative eigenspaces of $\mathcal{X}_{i}$, i.e. $\ES_{i} = \ES_{i}^{-}\oplus \ES_{i}^{+}$, with respect to which $\mathcal{X}_{i} = -\identity_{\ES_{i}^{-}}\oplus \identity_{\ES_{i}^{+}}$.  The dimension of $\ES_{i}^{-}$ will be denoted $\nu_{i}$, and the dimension of $\ES_{i}^{+}$ is then $n_{i}-\nu_{i}$.  This is equivalent to the statement that $\Tr(\mathcal{X}_{i}) = n_{i}-2\nu_{i}$.

As a result, the set of critical points of $J_{F}$ is 
given by
\begin{subequations}
\begin{align}
	\Crit(J_{F}) & = \Big\{U = W\bigoplus_{i=0}^{\kappa}\mathcal{X}_{i} \;:\; \mathcal{X}_{i}\in\mathrm{U}(\ES_{i})\text{ for } 0\leq i\leq\kappa\nonumber\\
	& \qquad  \text{ and } \mathcal{X}_{i}^{2} = \identity_{\ES_{i}} \text{ for } 1\leq i\leq\kappa\Big\}\\
	& = \bigsqcup_{\jmdstack{0\leq \nu_{i}\leq n_{i}}{\forall i=1,\dots,\kappa}} C_{\{\nu_{i}\}}
\end{align}	
\end{subequations}
where the critical submanifold $C_{\{\nu_{i}\}}$ is isomorphic to 
\begin{equation}
	\mathrm{U}(\ES_{0})\oplus\GrES{\nu_{1}}{1}\oplus\cdots\oplus\GrES{\nu_{\kappa}}{\kappa}
\end{equation}
and where $\mathcal{X}_{i}\in \mathrm{U}(\ES_{i})$ such that $\mathcal{X}_{i}^{2} = \identity_{\ES_{i}}$ is uniquely identified by its $-1$ eigenspace $\ES_{i}^{-}$, so the space of all such $\mathcal{X}_{i}$ with $\nu_{i}$-dimensional $-1$ eigenspace is isomorphic to the space of all $\nu_{i}$-dimensional subspaces of $\ES_{i}$, which is the Grasmannian 
\begin{equation}
	\GrES{\nu_{i}}{i}\simeq \frac{\mathrm{U}(\ES_{i})}{\mathrm{U}(\ES_{i}^{-})\oplus\mathrm{U}(\ES_{i}^{+})}.
\end{equation}
For each $\nu_{i}=0,\dots,n_{i}$, the Grassmannian of admissible $\mathcal{X}_{i}$ forms a connected submanifold of $\mathrm{U}(\ES_{i})$, and since the traces of the $\mathcal{X}_{i}$ corresponding to different $\nu_{i}$ are different [$\Tr(\mathcal{X}_{i}) = n_{i}-2\nu_{i}$] the Grasmannians on $\ES_{i}$ corresponding to different $\nu_{i}$ are disconnected.  With $n_{i}+1$ disjoint choices for each $i=1,\dots,\kappa$, it is clear that $J_{F}$ admits exactly $M = \prod_{i=1}^{\kappa}(n_{i}+1)$ connected critical submanifolds.

\subsection{Hessian Analysis}
\label{sec:JFHessian}
Turning to the question of the signatures of these critical points, we extend the gradient vector field $\grad J_{F}$ in the obvious way to all of $\BH$ and differentiate to find $\rmd_{U}\grad J_{F}(\delta U) = \delta U\, \AAdag W^{\dag}U + U \AAdag W^{\dag}\delta U$.  Projecting this onto the tangent bundle of $\UH$ gives the Hessian operator at $U$, $\Hess_{J_{F},U}\in\mathcal{B}\big(\rmT_{U}\UH\big)$,
\begin{subequations}
\begin{align}
	\Hess_{J_{F},U}(\delta U) & := \nabla_{\delta U}\grad J_{F}(U)\\
	& = \frac{1}{2}\big(\delta U \AAdag W^{\dag}U + W \AAdag U^{\dag}\delta U \nonumber\\
	& \qquad + U \AAdag W^{\dag}\delta U + \delta U\,U^{\dag}W \AAdag \big),
\end{align}
\end{subequations}
where $\nabla_{\delta U}$ denotes the covariant derivative in the direction $\delta U$ \cite{doCarmo1992}, and where we have used the fact that any tangent vector $\delta U\in \rmT_{U}\UH$ satisfies $\delta U^{\dag} = - U^{\dag}\delta U\,U^{\dag}$.  At a critical point, $\grad J_{F}(U) = 0$, so that $W^{\dag}U$ and $U^{\dag}W$ both commute with $ \AAdag $, and also $W^{\dag}U \AAdag = U^{\dag}W \AAdag $.  Then the Hessian becomes
\begin{equation}	
	\Hess_{J_{F},U}(\delta U) = \delta U \AAdag W^{\dag}U + U \AAdag W^{\dag}\delta U.
\end{equation}
Suppose that $U\in\UH$ is a critical point of $J_{F}$, let $\mathcal{X} = W^{\dag}U$ as before, and let $Y = U^{\dag}\delta U\in\uH$, where $\uH$ denotes the Lie algebra of skew-Hermitian operators on $\mathcal{H}$.  Then $\delta U$ is an eigenvector of $\Hess_{J_{F},U}$, i.e. $\Hess_{J_{F},U}(\delta U) = \lambda \delta U$, if and only if 
\begin{equation}
	Y \AAdag \mathcal{X} + \AAdag \mathcal{X} Y = \lambda Y\label{eqn:JFskewHermHessian}.
\end{equation}

We will use again the decomposition $\mathcal{H}=\bigoplus_{i=0}^{\kappa}\ES_{i}$ into eigenspaces of $ \AAdag $ and the further decomposition at a critical point $U=W\bigoplus\mathcal{X}_{i}$ of $\ES_{i} = \ES_{i}^{-}\oplus \ES_{i}^{+}$ into eigenspaces of $\mathcal{X}_{i}$.  The space $\uH$ may similarly be decomposed into subspaces of skew-Hermitian operators supported on ``diagonal'' and ``off-diagonal'' blocks
\begin{align}
	\uH & = \mathrm{u}(\ES_{0})\oplus\bigoplus_{i=1}^{\kappa}\mathrm{u}(\ES_{i}^{-})\oplus\mathrm{u}(\ES_{i}^{+})\oplus\frac{\mathrm{u}(\ES_{i})}{\mathrm{u}(\ES_{i}^{-})\oplus\mathrm{u}(\ES_{i}^{+})}\nonumber\\
	& \qquad \oplus\bigoplus_{0\leq i< j}^{\kappa}\bigoplus_{\jmdstack{s_{i}=\pm}{s_{j}=\pm}}\frac{\mathrm{u}(\ES_{i}^{s_{i}}\oplus \ES_{j}^{s_{j}})}{\mathrm{u}(\ES_{i}^{s_{i}})\oplus \mathrm{u}(\ES_{j}^{s_{j}})},
\end{align}
where the sum over $s_{i}$ is neglected for $i=0$ since $\ES_{0}$ is not decomposed into $\pm$ subspaces.
It is straightforward to see that the elements of these identified subspaces are eigenvectors of \eqref{eqn:JFskewHermHessian} with eigenvalues as in Table \ref{tab:JFHessSubspaces}.
\begin{table}
	\centering
	\caption{Eigenvalues and eigenspaces of the operator $\mathcal{L}_{U^{\dag}}\circ\Hess_{J_{F},U}\circ\mathcal{L}_{U}\in\mathcal{B}\big(\uH\big)$ described in \eqref{eqn:JFskewHermHessian}, where $\mathcal{L}_{U}$ denotes left multiplication by $U$.}
	\setlength{\extrarowheight}{3pt}
	\begin{tabular}{l|l}
	Subspace of $\mathrm{u}(\mathcal{H})$ & Eigenvalue $\lambda$\\
	\hline
	$\mathrm{u}(\ES_{0})$ & 0\\
	$\mathrm{u}(\ES_{i}^{\pm})$ & $\pm 2\omega_{i}^{2}$\\
	$\mathrm{u}(\ES_{i})/\big[\mathrm{u}(\ES_{i}^{-})\oplus\mathrm{u}(\ES_{i}^{+})\big]$ & 0\\
	$\mathrm{u}(\ES_{i}^{s_{i}}\oplus \ES_{j}^{s_{j}})/\big[\mathrm{u}(\ES_{i}^{s_{i}})\oplus \mathrm{u}(\ES_{j}^{s_{j}})\big]$ & $s_{i}\omega_{i}^{2} + s_{j}\omega_{j}^{2}$
	\end{tabular}
	\setlength{\extrarowheight}{0pt}
	\label{tab:JFHessSubspaces}
\end{table}

Notice that for orthogonal subspaces $Q$ and $R$ of dimension $q$ and $r$, respectively, $\dim[\mathrm{u}(Q)]= q^{2}$ and $\dim\big[\mathrm{u}(Q\oplus R)/\big(\mathrm{u}(Q)\oplus\mathrm{u}(R)\big)\big] = (q+r)^{2} - q^{2} - r^{2} = 2qr$.  Then the Hessian null space is 
\begin{equation}
	E_{0} = U\left(\mathrm{u}(\ES_{0})\oplus\bigoplus_{i=1}^{\kappa}\frac{\mathrm{u}(\ES_{i})}{\mathrm{u}(\ES_{i}^{-})\oplus\mathrm{u}(\ES_{i}^{+})}\right)
\end{equation}
which is identical to the tangent space of the associated critical submanifold, and consequently has the same dimension 
\begin{equation}
	\mathcal{N}_{0} = n_{0}^{2} + 2\sum_{i=1}^{\kappa}\nu_{i}(n_{i}-\nu_{i}).
\end{equation}
Therefore $J_{F}$ is a Morse-Bott function for all $A$ matrices.  The negative Hessian eigenspace (i.e. the negative Hessian eigenbundle at $U$, see Figure \ref{fig:vecBund}), spanned by the Hessian eigenspaces with strictly negative eigenvalues, is
\begin{equation}
	E_{-} = U\left(\bigoplus_{i=1}^{\kappa}\mathrm{u}(\ES_{i}^{-})\oplus\bigoplus_{0\leq i<j}^{\kappa}\frac{\mathrm{u}(\ES_{i}\oplus \ES_{j}^{-})}{\mathrm{u}(\ES_{i})\oplus\mathrm{u}(\ES_{j}^{-})}\right)
\end{equation} 
which has dimension
\begin{equation}
	\mathcal{N}_{-} = \sum_{i=1}^{\kappa}\nu_{i}^{2} + 2\sum_{0\leq i<j}^{\kappa}n_{i}\nu_{j}.\end{equation}
Finally, the positive Hessian eigenspace is
\begin{equation}
	E_{+} = U\left(\bigoplus_{i=1}^{\kappa}\mathrm{u}(\ES_{i}^{+})\oplus\bigoplus_{0\leq i<j}^{\kappa}\frac{\mathrm{u}(\ES_{i}\oplus \ES_{j}^{+})}{\mathrm{u}(\ES_{i})\oplus\mathrm{u}(\ES_{j}^{+})}\right)
\end{equation} 
which has dimension
\begin{equation}
	\mathcal{N}_{+} = \sum_{i=1}^{\kappa}(n_{i}-\nu_{i})^{2} + 2\sum_{0\leq i<j}^{\kappa}n_{i}(n_{j}-\nu_{j}).
\end{equation}

It is easy to see that $\mathcal{N}_{0} + \mathcal{N}_{-} + \mathcal{N}_{+} = N^{2} = \dim\big(\UH\big)$ as expected.  Furthermore, we find that $\mathcal{N}_{+} = 0$ if and only if $\nu_{i} = n_{i}$ for all $i = 1,\dots,\kappa$, i.e. only at the global maximum $\{U=-W\big(\mathcal{X}_{0}\oplus\identity_{\mathcal{H}/\ES_{0}}\big)\;:\;\mathcal{X}_{0}\in\mathrm{U}(\ES_{0})\}$.  Likewise $\mathcal{N}_{-} = 0$ if and only if $\nu_{i} = 0$ for all $i = 1,\dots,\kappa$, i.e. only at the global minimum $\{U=W\big(\mathcal{X}_{0}\oplus\identity_{\mathcal{H}/\ES_{0}}\big)\;:\;\mathcal{X}_{0}\in\mathrm{U}(\ES_{0})\}$.  So, there are no local traps in the kinematic landscape, and the remaining $\prod_{i=1}^{\kappa}(n_{i}+1) - 2$ critical submanifolds are all saddles.

\subsection{Examples}
\begin{example}[$ \AAdag $ is a projection]
Suppose that $ \AAdag = AA^{\dag}$ is a projection, or equivalently, that $A$ is a partial isometry.  Then $\mathcal{H} = \ES_{0}\oplus \ES_{1}$, $\omega_{0}^{2} = 0$, and $\omega_{1}^{2} = 1$.  By Theorem \ref{thm:FrobeniusSummary}, $J_{F}$ admits exactly $n_{1}+1$ critical submanifolds, each isomorphic to $\mathrm{U}(\ES_{0})\oplus \GrES{\nu_{1}}{1}$ and having dimension $\mathcal{N}_{0} = n_{0}^{2} + 2\nu_{1}(n_{1}-\nu_{1})$ for $\nu_{1}=0\,\dots, n_{1}$.  The critical submanifold identified by index $\nu_{1}$ has critical value $J_{F} = 4\nu_{1}$ and the Hessian has eigenvalues in the set $\{-2,-1,0,+1,+2\}$. At a critical point $U\in C_{\nu_{1}}$, the negative Hessian eigenspace 
\begin{equation}
	E_{-} = E_{-2}\oplus E_{-1} = U\left[\mathrm{u}(\ES_{1}^{-})\oplus \frac{\mathrm{u}(\ES_{0}\oplus \ES_{1}^{-})}{\mathrm{u}(\ES_{0})\oplus \mathrm{u}(\ES_{1}^{-})}\right]
\end{equation}
is of dimension $\mathcal{N}_{-}=\nu_{1}^{2} + 2n_{0}\nu_{1}$, and the positive Hessian eigenspace
\begin{equation}
	E_{+} = E_{+2} \oplus E_{+1} = U\left[\mathrm{u}(\ES_{1}^{+})\oplus \frac{\mathrm{u}(\ES_{0}\oplus \ES_{1}^{+})}{\mathrm{u}(\ES_{0})\oplus \mathrm{u}(\ES_{1}^{+})}\right]
\end{equation}
is of dimension $\mathcal{N}_{+} = (n_{1}-\nu_{1})^{2} + 2n_{0}(n_{1}-\nu_{1})$.  The full set of critical values for this problem is $\{0,4,8,\dots, 4n_{1}\}$.
\end{example}

\begin{example}[$ \AAdag $ is nondegenerate and nonsingular]
Suppose that $ \AAdag $ is nondegenerate and nonsingular, so that $\mathcal{H} = \bigoplus_{i=1}^{N} \ES_{i}$, $n_{i}=1$ for all $i=1,\dots, N$, and $\omega_{1}^{2} < \dots < \omega_{N}^{2}$.  By Theorem \ref{thm:FrobeniusSummary}, the critical set of $J_{F}$ comprises exactly $2^{N}$ isolated critical points (i.e. zero-dimensional critical submanifolds).  For the critical point identified by indices $(\nu_{1},\dots,\nu_{N})$, the critical value is $J_{F} = 4\sum_{i=1}^{N}\nu_{i}\omega_{i}^{2}$, the negative Hessian eigenbundle has dimension $\mathcal{N}_{-} = \sum_{i=1}^{N}(2i-1)\nu_{i}$, and the positive eigenbundle has dimension $\mathcal{N}_{+} = N^{2} - \sum_{i=1}^{N}(2i-1)\nu_{i}$.
\end{example}

\bigskip

\section{Kinematic Critical Point Analysis of Certain Phase-Invariant Landscapes}\label{sec:phaseInvKinematic}
We now turn our attention to the kinematic landscape 
\begin{equation}
	J_{P}(U) := \|A\|^{4} - |\Tr( \AAdag W^{\dag}U)|^{2}.
\end{equation}
This function is phase-invariant, meaning that $J_{P}(e^{i\theta}U) = J_{P}(U)$ for any $\theta\in\mathbb{R}$.  Since the global phase of a state vector $|\psi\rangle\in\mathcal{H}$ has no physical meaning, neither does the global phase of the unitary propagator, so that $U$ and $e^{i\theta}U$ are functionally equivalent.  A phase-invariant objective function such as $J_{P}$ which treats such equivalent operators as equally optimal may therefore be desirable as it may be expected to require optimization only with respect to the degrees of freedom that are physically relevant.

\begin{figure}
	\begin{equation}
		\xymatrixcolsep{0pt}\xymatrix{ \mathrm{U}(1) \ar[rrrrr]^{\phase\mapsto \phase\identity} & & & & &\UH\ar[rd]^{\pi} & & \SUH \ar[ld]_{p} & & & & & & & & & \mathbb{Z}/N\mathbb{Z} \ar[lllllllll]_{e^{i2\pi k/N}\identity\mapsfrom k}\\  & & & & & & \PUH}\nonumber
	\end{equation}
	\caption{$\UH$ and $\SUH$ as fibre bundles over $\PUH$.  The two compositions of maps depicted here are exact, i.e. each composition $f\circ g$ is such that $\ker(f) = \im(g)$.  Since $\mathrm{U}(1)\identity\subset \UH$ is the subgroup of global phase rotations, $\PUH$ may be thought of as the unitary group modulo global phase: two unitary operators $U,W\in \UH$ will be mapped by $\pi$ to the same operator in $\PUH$ if and only if $W^{\dagger}U\in\mathrm{U}(1)\identity$, i.e., if and only if $U = e^{i\theta}W$ for some $\theta\in\mathbb{R}$.}
	\label{fig:commdiag}
\end{figure}

The remainder of Section \ref{sec:phaseInvKinematic} will be concerned with proving the following theorem.
\begin{theorem}
	The critical set of the kinematic objective function $J_{P}$ comprises a global maximum set and $M$ connected nondegenerate critical submanifolds.  The global maximum set need not globally be a submanifold of $\UH$, but away from self-intersection points is a codimension 2 submanifold of $\UH$.  The remaining critical submanifolds are of the form
\begin{align}
	C_{\{\nu_{i}\}} & = \Big\{U = \phase WZ\,:\,\phase \in\mathrm{U}(1), Z\in\bigoplus_{i=0}^{\kappa}\mathrm{U}(\ES_{i}), \text{ and }\nonumber\\
	& \quad Z_{i}^{2} = \identity_{\ES_{i}} \text{ with } \Tr(Z_{i}) = n_{i}-2\nu_{i} \text{ for }i\geq 1\Big\}\\
	& \simeq \mathrm{U}(\ES_{0})\oplus \left[U(1)\times \bigoplus_{i=1}^{\kappa}\GrES{\nu_{i}}{i}\right],
\end{align} 
where $0\leq \nu_{i}\leq n_{i}$ are such that $\sum_{i}(n_{i}-2\nu_{i})\omega_{i}^{2} > 0$.  The number, $M$, of these critical submanifolds is equal to the number of choices of these integers $\{\nu_{i}\}$ satisfying the above two conditions and therefore depends on the singular values $\{\omega_{i}\}$ of the parameter operator $A$.  The critical submanifold $C_{\{\nu_{i}\}}$ described by a particular set of indices $(\nu_{1},\dots,\nu_{\kappa})$ corresponds to a critical value of 
\begin{equation} 
	J_{P} = 4\left(\sum_{i=1}^{\kappa}\nu_{i}\omega_{i}^{2}\right)\left(\sum_{i=1}^{\kappa}(n_{i}-\nu_{i})\omega_{i}^{2}\right)
\end{equation}
and has dimension 
\begin{equation}
	\mathcal{N}_{0} = 1+ n_{0}^{2} + 2\sum_{i=1}^{\kappa}\nu_{i}(n_{i}-\nu_{i}),
\end{equation}
while the dimensions of the negative and positive Hessian eigenbundles on this submanifold are 
\begin{align}
	\mathcal{N}_{-} & = \sum_{i=1}^{\kappa}\nu_{i}^{2} + 2\sum_{0\leq i<j}^{\kappa}n_{i}\nu_{j}\\
	\mathcal{N}_{+} & = -1 + \sum_{i=1}^{\kappa}(n_{i}-\nu_{i})^{2} + 2\sum_{0\leq i<j}^{\kappa}n_{i}(n_{j}-\nu_{j}).
\end{align}
Consequently, of these submanifolds $C_{\{\nu_{i}\}}$, exactly one (corresponding to the case $\nu_{i}=0$ for all $i=1,\dots,\kappa$) is the set of global minima $\{U=\phase W(Z_{0}\oplus\identity_{\mathcal{H}/\ES_{0}})\;:\; \phase \in\mathrm{U}(1) \text{ and } Z_{0}\in\mathrm{U}(\ES_{0})\}$, and the remaining $M-1$ critical submanifolds are all saddles, so that $J_{P}$ admits no local traps.  For an open, dense set of $A$ operators in $\BH$, the global maximum set of $J_{P}$ is a nondegenerate submanifold of $\UH$, in which case $J_{P}$ is a Morse-Bott function.
\label{thm:phaseInvariantSummary}
\end{theorem}

\subsection{Distance Metric on \texorpdfstring{$\PUH$}{PU(H)}}
There are various ways of deriving a phase-invariant landscape like $J_{P}$ from one that is phase-dependent like $J_{F}$.  A simple approach is to observe that 
\begin{equation}
	\min_{\phase\in\mathrm{U}(1)}\|(U - \phase W)A\|^{2} = 2\|A\|^{2} - 2|\Tr( \AAdag W^{\dag}U)|.
\end{equation}
  This provides a means to define a quotient metric on the projective unitary group $\PUH$ (see Figure \ref{fig:commdiag}) from the metric $d(U,W) = \|(U-W)A\|^{2}$ on $\UH$.  

Another approach involves the adjoint representation of the unitary group, $\Ad:\UH\rightarrow \Aut\big(\uH\big)\subset \mathrm{GL}\big(\uH\big)\cong \mathrm{GL}(N^{2}; \mathbb{R})$, which is given by $\Ad(U)A = U A U^{\dag}$ for any $A\in \uH$, and where $\Aut\big(\uH\big)$ is the group of Lie algebra automorphisms on $\uH$ \cite{Warner1983, Knapp2004}.  With $\uH$ given the Hilbert-Schmidt inner product, $\langle \Ad(U)A, \Ad(U)B\rangle = \langle U A U^{\dag}, U B U^{\dag}\rangle  = \langle A, B\rangle$, so that for each $U\in\UH$, $\Ad(U)$ is an orthogonal operator on $\uH$, i.e. $\Ad:\UH\to\mathrm{SO}(\uH)$.  Furthermore, if $W$ and $U$ differ only by a global phase, i.e. $W = \phase U$, then $\Ad(U) = \Ad(W)$.  Moreover, the kernel of $\Ad$, i.e. $\Ad^{-1}(\rm{id})$, is the center of $\UH$ \cite[Cor. 5.2, pg. 129]{Helgason2001}\cite[Thm. 3.50]{Warner1983} which is $Z\big(\UH\big) = \{\phase \identity\} = \mathrm{U}(1)\identity$.  Then $\Ad(U)= \Ad(W)$ if and only if $W = \phase U$, so the image of $\Ad$ is a faithful representation of $\PUH$ and $\Ad$ may be thought of as playing a role similar to the projection $\pi:\UH\to\PUH$.

Consider some target $W\in\UH$ and some $B\in \mathrm{GL}\big(\uH)$ and define 
\begin{equation}
	J(U) = \frac{1}{2}\|(\Ad(U)-\Ad(W))\circ B\|_{\mathrm{HS}}^{2},
\end{equation}
where $\|\cdot\|_{\mathrm{HS}}$ is the Hilbert-Schmidt norm on $\End\big(\uH\big)\cong \mathbb{R}^{N^{2}\times N^{2}}$, the space of all linear operators acting on $\uH$.  Then, 
\begin{align}
	J(U) & = \Tr(B^{*} B) - \Tr\big(B^{*}\Ad(W)^{*}\Ad(U) B\big). \label{eqn:genAdjointDist}
\end{align}
Let $A\in \BH$ be an arbitrary linear operator on $\mathcal{H}$, and let $B$ be defined by $B(\Omega) = A\Omega A^{\dag}$.  It follows from \eqref{eqn:genAdjointDist} that the kinematic landscape $J_{P}(U) = \|A\|^{4} - |\Tr( \AAdag W^{\dag}U)|^{2}$ on $\UH$ is equivalent to the weighted Hilbert-Schmidt distance function on the subgroup of $\mathrm{SO}\big(\uH\big)$ given by $\mathrm{Im}(\Ad)\simeq \PUH$.  In other words, $J_{P}$ is completely equivalent to $J_{F}$, but applied to $\PUH$, rather than $\UH$.

\subsection{Critical Point Identification}
Now, the differential of $J_{P}$ at $U\in\UH$, $\rmd_{U}J_{P}:\rmT_{U}\UH\rightarrow \mathbb{R}$ is given by 
\begin{subequations}
\begin{align}
	\rmd_{U}J_{P}(\delta U) & = -\Tr( \AAdag W^{\dag}\delta U)\Tr(U^{\dag}W \AAdag )\nonumber\\
	& \qquad  - \Tr( \AAdag W^{\dag}U)\Tr(\delta U^{\dag}W \AAdag )\\
	& = \big\langle\Tr(U^{\dag}W \AAdag )U \AAdag W^{\dag}U\nonumber\\
	& \qquad \qquad  - \Tr( \AAdag W^{\dag}U)W \AAdag , \delta U \big\rangle
\end{align}
\end{subequations}
so that 
\begin{equation}
	\grad J_{P}(U) = \Tr(U^{\dag}W \AAdag )U \AAdag W^{\dag}U - \Tr( \AAdag W^{\dag}U)W \AAdag .\label{eq:gradJP}
\end{equation}
For $J_{P}(U)<\|A\|^{4}$, $\Tr(U^{\dag}W \AAdag )\neq 0$, so $\grad J_{P}(U) = 0$ if and only if $ \AAdag Z = Z^{\dag} \AAdag $, where 
\begin{equation}
	Z = \frac{\Tr(U^{\dag}W \AAdag )}{|\Tr(U^{\dag}W \AAdag )|}W^{\dag}U\in \UH.
\end{equation}
This same condition was considered in section \ref{sec:frobkinematic} (and \cite{Ho2009}), where it was shown to imply that, under the orthogonal decomposition $\mathcal{H}=\bigoplus_{i=0}^{\kappa}\ES_{i}$ of $\mathcal{H}$ into the eigenspaces of dimensions $\{n_{i}\}$ of $ \AAdag $, $Z = \bigoplus_{i=0}^{\kappa}Z_{i}$ with $Z_{i}\in\mathrm{U}(\ES_{i})$ for $i=0,\dots,\kappa$ and $Z_{i}^{2} = \identity_{\ES_{i}}$ for $i\geq 1$.  Since the involutions $Z_{i}$ for $i\geq 1$ have eigenvalues $\pm1$, they induce a further orthogonal decomposition of $\ES_{i}$ into $\ES_{i} = \ES_{i}^{-}\oplus \ES_{i}^{+}$ into the $\pm 1$ eigenspaces of $Z_{i}$ of dimensions $\nu_{i}$ and $n_{i}-\nu_{i}$, respectively.

Drawing on the material above, we find that any critical point $U$ of $J_{P}$ with $J_{P}(U)<\|A\|^{4}$ can be written as 
\begin{equation}
	U = \frac{\Tr( \AAdag W^{\dag}U)}{|\Tr(U^{\dag}W \AAdag )|}WZ
\end{equation}
with $Z = \bigoplus_{i=0}^{\kappa} Z_{i}$ and $Z_{i}^{2} = \identity_{\ES_{i}}$ for $i\geq 1$.  This characterization is complicated by the presence of $U$ on both sides of the equation, especially with regard to the phase factor on the right hand side.  However, it may be observed for any $U = \phase WZ$ with $\phase\in\mathrm{U}(1)$, $Z\in\bigoplus_{i=0}^{\kappa}\mathrm{U}(\ES_{i})$, and $Z_{i}^{2}=\identity_{\ES_{i}}$ for $i\geq 1$, that $[U^{\dag}W, \AAdag ] = [W^{\dag}U, \AAdag ] = 0$ and $U^{\dag}W \AAdag = ( \AAdag W^{\dag}U)^{\dag} = \phase^{-2} \AAdag W^{\dag}U$, so that $\grad J_{P}(U) = 0$.  Hence, \emph{every} such $U$ is a critical point of $J_{P}$, and they comprise connected critical sets 
\begin{align}
	C_{\{\nu_{i}\}} & := \Big\{U = \phase WZ\,:\,\phase \in\mathrm{U}(1), Z\in\bigoplus_{i=0}^{\kappa}\mathrm{U}(\ES_{i}), \text{ and }\nonumber\\
	& \qquad Z_{i}^{2} = \identity_{\ES_{i}} \text{ with } \dim(\ES_{i}^{-}) = \nu_{i}\text{ for } i\geq 1\Big\}
\end{align}
for all $0\leq \nu_{i}\leq n_{i}$.  However, it may be observed that for any such set of indices $\{\nu_{i}\}$, $C_{\{\nu_{i}\}} = C_{\{n_{i}-\nu_{i}\}}$, since $U = \phase WZ\in C_{\{\nu_{i}\}}$ if and only if $U = (-\phase )W(-Z)\in C_{\{n_{i}-\nu_{i}\}}$.  It suffices then to only consider $C_{\{\nu_{i}\}}$ for which $\Tr( \AAdag Z) = \sum_{i=1}^{\kappa}\omega_{i}^{2}(n_{i}-2\nu_{i})>0$ to avoid identifying the same critical submanifold twice.  Such a critical submanifold $C_{\{\nu_{i}\}}$ has the critical value 
\begin{subequations}
\begin{align}
	J_{P}(U) & = \big(\Tr\big( \AAdag \big)\big)^{2} - |\Tr\big( \AAdag W^{\dag}U\big)|^{2}\\
	& = \left(\sum_{i=1}^{\kappa}n_{i}\omega_{i}^{2}\right)^{2} - \left(\sum_{i=1}^{\kappa}(n_{i}-2\nu_{i})\omega_{i}^{2}\right)^{2}\\
	& = 4\left(\sum_{i=1}^{\kappa}\nu_{i}\omega_{i}^{2}\right)\left(\sum_{i=1}^{\kappa}(n_{i}-\nu_{i})\omega_{i}^{2}\right)
\end{align}
\end{subequations}
for every $U\in C_{\{\nu_{i}\}}$.

\subsection{Hessian Analysis}
Given the form of the gradient of $J_{P}$ in \eqref{eq:gradJP}, by again extending the gradient vector field to $\BH$ and differentiating, it is found that
\begin{align}
	&\rmd_{U}\grad J_{P}(\delta U)\nonumber\\
	& = \Tr(\delta U^{\dag}W \AAdag )U \AAdag W^{\dag}U  + \Tr(U^{\dag}W \AAdag )\delta U \AAdag W^{\dag}U\nonumber\\
	& \qquad + \Tr(U^{\dag}W \AAdag )U \AAdag W^{\dag}\delta U - \Tr( \AAdag W^{\dag}\delta U)W \AAdag ,
\end{align}
whence, by projection onto the tangent bundle of $\UH$,
\begin{align}
	& \Hess_{J_{P},U}(\delta U) = \nabla_{\delta U}\grad J_{P}\nonumber\\
	& = -\Tr( \AAdag W^{\dag}\delta U)W \AAdag + \Tr(\delta U^{\dag}W \AAdag )U \AAdag W^{\dag}U\nonumber\\
	& \quad  + \!\frac{1}{2}\Big\{\Tr(U^{\dag}W \AAdag )\delta U \AAdag W^{\dag}U \!+\! \Tr( \AAdag W^{\dag}U) W \AAdag U^{\dag}\delta U\nonumber\\
	& \quad \; + \Tr(U^{\dag}W \AAdag )U \AAdag W^{\dag}\delta U + \Tr( \AAdag W^{\dag}U)\delta U U^{\dag} W \AAdag \Big\}.
\end{align}
On one of the critical submanifolds $C_{\{\nu_{i}\}}$, the Hessian is given by 
\begin{align}
	\Hess_{J_{P},U}(\delta U) & =  \Tr( \AAdag W^{\dag}U)\delta U U^{\dag}W \AAdag \nonumber\\
	&\qquad + \Tr( \AAdag W^{\dag}U)W \AAdag U^{\dag}\delta U\nonumber\\
	& \qquad\mbox{} - 2\Tr( \AAdag W^{\dag}\delta U)W \AAdag .
\end{align}
Writing a critical $U\in C_{\{\nu_{i}\}}$ as $U = \phase WZ$ and letting $Y = U^{\dag}\delta U$, the Hessian eigenvalue problem $\Hess_{J_{P},U}(\delta U ) = \lambda \delta U$ can be written as an eigenvector problem on $\uH$ as
\begin{equation}
	\Tr( \AAdag Z)\big[YZ \AAdag + Z \AAdag Y\big] - 2\Tr( \AAdag ZY)Z \AAdag = \lambda Y.
	\label{eqn:JPskewHermHessian}	
\end{equation}
Observe that $\Tr( \AAdag ZY) = 0$ for any $Y$ in $\mathrm{u}(\ES_{0})$, $\mathrm{su}(\ES_{i}^{\pm})$ for $i=1,\dots, \kappa$, $\mathrm{u}(\ES_{i})/\big(\mathrm{u}(\ES_{i}^{-})\oplus\mathrm{u}(\ES_{i}^{+})\big)$ for $i=1,\dots, \kappa$, or $\mathrm{U}(\ES_{i}^{s_{i}}\oplus \ES_{j}^{s_{j}})/\big[\mathrm{U}(\ES_{i}^{s_{i}})\oplus\mathrm{U}(\ES_{j}^{s_{j}})\big]$ for $0\leq i< j$ and $s_{i},s_{j}\in\{\pm\}$.  So for $Y$ in any of these 
subspaces of $\uH$, the eigenvalue problem \eqref{eqn:JPskewHermHessian} becomes
\begin{equation}
	\Tr( \AAdag Z)\big[YZ \AAdag + Z \AAdag Y\big] = \lambda Y,
\end{equation}
which means, as in Section \ref{sec:JFHessian}, that each element of these subspaces is an eigenvector as in Table \ref{tab:JPHessSubspaces}.
\begin{table}
	\centering
	\caption{Some eigenvalues and eigenspaces of the operator $\mathcal{L}_{U^{\dag}}\circ\Hess_{J_{P},U}\circ\mathcal{L}_{U}\in\mathcal{B}\big(\uH\big)$ described in \eqref{eqn:JPskewHermHessian}, where $\mathcal{L}_{U}$ denotes left multiplication by $U$.}
	\setlength{\extrarowheight}{3pt}
	\begin{tabular}{l|l}
	Subspace of $\mathrm{u}(\mathcal{H})$ & Eigenvalue $\lambda$\\
	\hline
	$\mathrm{u}(\ES_{0})$ & 0\\
	$\mathrm{su}(\ES_{i}^{\pm})$ & $\pm 2\omega_{i}^{2}\Tr( \AAdag Z)$\\
	$\mathrm{u}(\ES_{i})/\big[\mathrm{u}(\ES_{i}^{-})\oplus\mathrm{u}(\ES_{i}^{+})\big]$ & 0\\
	$\mathrm{u}(\ES_{i}^{s_{i}}\oplus \ES_{j}^{s_{j}})/\big[\mathrm{u}(\ES_{i}^{s_{i}})\oplus \mathrm{u}(\ES_{j}^{s_{j}})\big]$ & $(s_{i}\omega_{i}^{2} + s_{j}\omega_{j}^{2})\Tr( \AAdag Z)$
	\end{tabular}
	\setlength{\extrarowheight}{0pt}
	\label{tab:JPHessSubspaces}
\end{table}
The only subspace of $\uH$ not covered by these cases is the subspace $S\subset\uH$ spanned by elements $Y\in \uH$ of the form $Y = \bigoplus_{i=1}^{\kappa}\big(\alpha_{i}^{-}\identity_{\ES_{i}^{-}} \oplus \alpha_{i}^{+}\identity_{\ES_{i}^{+}}\big)$ for imaginary numbers $\{\alpha_{i}^{\pm}\}$.  For such a $Y$, \eqref{eqn:JPskewHermHessian} is block diagonal with $\ES_{i}^{\pm}$ diagonal block 
\begin{equation}
	\pm\omega_{i}^{2}\big(2\Tr( \AAdag Z)\alpha_{i}^{\pm} - 2\Tr( \AAdag ZY)\big)\identity_{\ES_{i}^{\pm}} = \lambda \alpha_{i}^{\pm}\identity_{\ES_{i}^{\pm}}.
\end{equation}
Solving for $\alpha_{i}^{\pm}$, we find that, for $\lambda\neq\pm\omega_{i}^{2}\Tr( \AAdag Z)$ for all $\ES_{i}^{\pm}$ such that $\dim(\ES_{i}^{\pm})>0$, 
\begin{equation}
	\alpha_{i}^{\pm} = \frac{2\omega_{i}^{2}\Tr( \AAdag ZY)}{2\omega_{i}^{2}\Tr( \AAdag Z) \mp\lambda}.
	\label{eqn:JPremEigenvectorSoln}
\end{equation}
Then for $Y = \bigoplus_{i=1}^{\kappa}\big(\alpha_{i}^{-}\identity_{\ES_{i}^{-}} \oplus \alpha_{i}^{+}\identity_{\ES_{i}^{+}}\big)\in S$, it follows that 
\begin{subequations}
\begin{align}
	\Tr&( \AAdag ZY) = \sum_{i=1}^{\kappa}\omega_{i}^{2}(-\nu_{i}\alpha_{i}^{-} + (n_{i}-\nu_{i})\alpha_{i}^{+})\\
	& = \Tr( \AAdag ZY)\sum_{i=1}^{\kappa}2\omega_{i}^{4}\Bigg(-\frac{\nu_{i}}{2\omega_{i}^{2}\Tr( \AAdag Z) + \lambda}\nonumber\\
	& \qquad \qquad \qquad \qquad \qquad  + \frac{n_{i}-\nu_{i}}{2\omega_{i}^{2}\Tr( \AAdag Z) - \lambda}\Bigg)
\end{align}
\end{subequations}
which implies that either $\Tr( \AAdag ZY) = 0$ or $f(\lambda) = 1$ where
\begin{equation}
	f(\lambda) := \sum_{i=1}^{\kappa}\left(-\frac{2\omega_{i}^{4}\nu_{i}}{2\omega_{i}^{2}\Tr( \AAdag Z) + \lambda} + \frac{2\omega_{i}^{4}(n_{i}-\nu_{i})}{2\omega_{i}^{2}\Tr( \AAdag Z) - \lambda}\right).
	\label{eqn:flambda}
\end{equation}

\begin{figure}
	\centering
	\includegraphics[scale=0.9,clip = true]{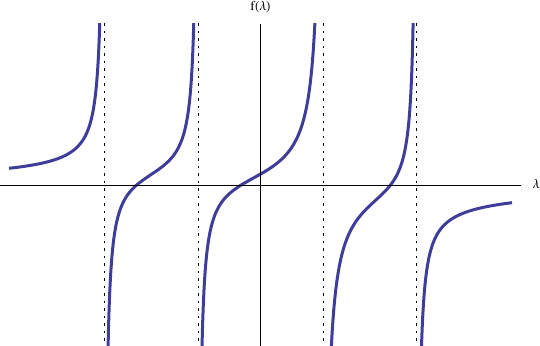}
	\caption{A depiction of the general properties of $f(\lambda)$ defined in Eq.\eqref{eqn:flambda}, namely that $f(\lambda)\to 0$ as $\lambda\to\pm\infty$, $f(0) = 1$, $f$ exhibits  simple poles at some subset of $\{\pm 2\omega_{i}^{2}\Tr(\AAdag Z)\}$, and that $f$ is stricly increasing away from these poles.}
	\label{fig:flambda}
\end{figure}

It may be observed that $f(\lambda)\to 0$ as $\lambda\to\pm\infty$, that $f$ is an increasing function away from its poles, and that $f$ has a simple pole at $-2\omega_{i}^{2}\Tr( \AAdag Z)$ for any $i=1,\dots,\kappa$ for which $\nu_{i}>0$, and a simple pole at $2\omega_{i}^{2}\Tr( \AAdag Z)$ for any $i$ for which $\nu_{i}<n_{i}$ (see Fig. \ref{fig:flambda}).  The number of distinct poles is then equal to the dimension of the subspace $S\in\uH$ under consideration; furthermore, if $\nu_{i}=n_{i}$ for all $i$, then $\Tr( \AAdag Z)<0$ which was disallowed by convention, so $f$ must have at least one positive pole.  It is then clear that $f(\lambda) = 1$ has $\dim(S)$ distinct solutions: one less than the smallest pole, and one between each pair of adjacent poles.  Moreover, it may be seen that $\lambda = 0$ is one of these solutions, corresponding to the eigenvector $Y^{(0)} = i\bigoplus_{i=1}^{\kappa}\identity_{\ES_{i}}$ of \eqref{eqn:JPskewHermHessian}.  The number of  solutions $f(\lambda) = 1$ with $\lambda < 0$ is then equal to the number $\nu_{+}$ of $i\in\{1,\dots,\kappa\}$ for which $\nu_{i}>0$, and the number of solutions $f(\lambda) = 1$ with $\lambda > 0$ is one less than the number $\nu_{-}$ of $i\in\{1,\dots,\kappa\}$ for which $\nu_{i}<n_{i}$.  To each of these solutions $\lambda$, we may associate the eigenvector 
\begin{equation}
	Y^{(\lambda)} = \bigoplus_{i=1}^{\kappa}\big(\alpha_{i}^{-}\identity_{\ES_{i}^{-}} \oplus \alpha_{i}^{+}\identity_{\ES_{i}^{+}}\big)\in S
\end{equation}
with 
\begin{equation}
	\alpha_{j}^{\pm} = \frac{2i \omega_{j}^{2}\Tr( \AAdag Z)}{2\omega_{j}^{2}\Tr( \AAdag Z) \mp \lambda}.
\end{equation}

So, to summarize, at a critical point $U = \phase WZ\in C_{\{\nu_{i}\}}$, the Hessian null space is 
\begin{equation}
	E_{0} = U\left(\mathrm{u}(\ES_{0})\oplus \left[\mathrm{u}(1)\identity_{\mathcal{H}/\ES_{0}} \times\bigoplus_{i=1}^{\kappa}\frac{\mathrm{u}(\ES_{i})}{\mathrm{u}(\ES_{i}^{-})\oplus\mathrm{u}(\ES_{i}^{+})}\right]\right)
\end{equation}
which is readily seen to be identical to the tangent space of $C_{\{\nu_{i}\}}$ and is of dimension
\begin{equation}
	\mathcal{N}_{0} = 1 + n_{0}^{2} + 2\sum_{i=1}^{\kappa}\nu_{i}(n_{i}-\nu_{i}).
\end{equation}
The negative Hessian eigenspace is 
\begin{equation}
	E_{-} \!= U\!\left(\!\bigoplus_{\jmdstack{\lambda<0}{f(\lambda)=1}}\!\!\!\mathbb{R}Y^{(\lambda)}\!\oplus\bigoplus_{i=1}^{\kappa}\mathrm{su}(\ES_{i}^{-})\!\oplus\!\!\!\bigoplus_{0\leq i < j}^{\kappa}\!\frac{\mathrm{u}(\ES_{i}\oplus \ES_{j}^{-})}{\mathrm{u}(\ES_{i})\!\oplus\!\mathrm{u}(\ES_{j}^{-})}\!\!\right)\!,
\end{equation}
having dimension
\begin{equation}
	\mathcal{N}_{-} = \sum_{i=1}^{\kappa}\nu_{i}^{2} +2\sum_{0\leq i < j}^{\kappa}n_{i}\nu_{j},
\end{equation}
and the positive Hessian eigenspace is 
\begin{equation}
	E_{+} \!= U\!\left(\!\bigoplus_{\jmdstack{\lambda>0}{f(\lambda)=1}}\!\!\!\mathbb{R}Y^{(\lambda)}\!\oplus\bigoplus_{i=1}^{\kappa}\mathrm{su}(\ES_{i}^{+})\!\oplus\!\!\!\bigoplus_{0\leq i < j}^{\kappa}\!\frac{\mathrm{u}(\ES_{i}\oplus \ES_{j}^{+})}{\mathrm{u}(\ES_{i})\!\oplus\!\mathrm{u}(\ES_{j}^{+})}\!\!\right)\!,
\end{equation}
having dimension
\begin{equation}
	\mathcal{N}_{+} = -1 + \sum_{i=1}^{\kappa}(n_{i}-\nu_{i})^{2} +2\sum_{0\leq i < j}^{\kappa}n_{i}(n_{j}-\nu_{j}).
\end{equation}
As a result, $\mathcal{N}_{-} = 0$ if and only if $\nu_{i} = 0$ for all $i = 1,\dots,\kappa$, i.e. only at the global minimum
\begin{equation}
	\{U=\phase W(Z_{0}\oplus\identity_{\mathcal{H}/\ES_{0}})\;:\; \phase \in\mathrm{U}(1), Z_{0}\in\mathrm{U}(\ES_{0})\}.
\end{equation}
Furthermore, $\mathcal{N}_{+}$ can be zero for such a critical point only if $\sum_{i=1}^{\kappa}(n_{i}-\nu_{i})^{2} = 1$ and $\sum_{0\leq i < j}^{\kappa}n_{i}(n_{j}-\nu_{j}) = 0$, i.e. only if $n_{0} = 0$, $\nu_{1} = n_{1}-1$, and $\nu_{i} = n_{i}$ for $i=2,\dots,\kappa$.  However, in this case, the constraint $\Tr( \AAdag Z)>0$ implies that $2\omega_{1}^{2}>\sum_{i=1}^{\kappa}n_{i}\omega_{i}^{2}$, which can only happen in the trivial case $\kappa = N = 1$ where global phase rotations are the only dynamics and every point is critical with respect to $J_{P}$.  Therefore, there are no maxima among the $C_{\{\nu_{i}\}}$.  The only maximal points of $J_{P}$ belong to the global maximum set considered presently.

\subsection{Global Maximum Set}
Finally, consider the global maximum set $\{U\;:\; J_{P}(U) = \|A\|^{4}\} = \{U\;:\;\Tr( \AAdag W^{\dag}U) = 0\}$, which is the intersection of $\UH$ with the (complex) hyperplane in $\BH$ orthogonal to $W \AAdag$.  This set does not admit analysis by the methods used thus far, so a different approach is required.  Let $F:\UH\rightarrow \mathbb{R}^{2}$ be given by $F_{1}:= \Re\Tr( \AAdag W^{\dag}U)$ and $F_{2} := \Im\Tr( \AAdag W^{\dag}U)$.  Then,
\begin{subequations}
\begin{align}
	\rmd_{U}F_{1}(\delta U) & = \Re\Tr( \AAdag W^{\dag}\delta U) \nonumber\\
	& = \frac{1}{2}\langle W \AAdag - U \AAdag W^{\dag}U, \delta U\rangle\\
	\rmd_{U}F_{2}(\delta U) & = \Im\Tr( \AAdag W^{\dag}\delta U) \nonumber\\
	&  = \frac{1}{2}\langle iW \AAdag + iU \AAdag W^{\dag}U,\delta U\rangle
\end{align}
\end{subequations}
so that the gradients are given by
\begin{subequations}
\begin{align}
	\grad F_{1}(U) & = \frac{1}{2}\big(W \AAdag - U \AAdag W^{\dag}U\big)\\
	\grad F_{2}(U) & = \frac{i}{2}\big(W \AAdag + U \AAdag W^{\dag}U\big).
\end{align}
\end{subequations}
Thus, $\rmd F$ is surjective except where there exists $(\alpha, \beta)\neq 0\in\mathbb{R}^{2}$ such that $\alpha\grad F_{1}(U) = \beta\grad F_{2}(U)$, i.e. where
\begin{equation}\phase U^{\dag}W \AAdag = \bar{\phase} \AAdag W^{\dag}U,\end{equation}
where $\phase = (\alpha-i\beta)/|\alpha + i\beta|$.  As we have already seen, this equation implies that $U = \phase W\bigoplus_{i=0}^{\kappa}Z_{i}$ with $Z_{i}^{2}=\identity_{\ES_{i}}$ for $i\geq 1$. 

For such a $U$, 
\begin{equation}
	\Tr( \AAdag W^{\dag}U) = \phase\Tr( \AAdag Z) = \phase\sum_{i=1}^{\kappa}\omega_{i}^{2}(n_{i}-2\nu_{i}),
\end{equation}
so the only way that $\Tr( \AAdag W^{\dag}U)$ can be zero is if the vector $(\omega_{1}^{2},\dots,\omega_{\kappa}^{2})$ is orthogonal to one of the possible vectors $(n_{1}-2\nu_{1},\dots,n_{\kappa}-2\nu_{\kappa})$, i.e. $(\omega_{1}^{2},\dots,\omega_{\kappa}^{2})$ must lie in the union of the $2^{N-n_{0}-1}$ hyperplanes which are the orthogonal spaces of the vectors $(n_{1}-2\nu_{1},\dots,n_{\kappa}-2\nu_{\kappa})$.  Consequently, for a given orthogonal decomposition $\mathcal{H}=\bigoplus_{i=0}^{\kappa}\ES_{i}$,  the collection of all $A$'s for which $\mathcal{A}^{2} = AA^{\dagger} = \bigoplus_{i=1}^{\kappa}\omega_{i}^{2}\identity_{\ES_{i}}$ and $\rmd F$ is surjective at $F(U) = 0$ [hence $\{U\;:\;F(U) = 0\}$ is a codimension 2 submanifold of $\UH$] comprises an open dense subset of the $A$'s in $\BH$ for which $AA^{\dagger} = \bigoplus_{i=1}^{\kappa}\omega_{i}^{2}\identity_{\ES_{i}}$.  It follows that the set of all $A$ operators for which $\rmd F$ is surjective at $F(U) = 0$ is open and dense in $\BH$.

Now, at a point $U$ such that $\Tr( \AAdag W^{\dag}U) = 0$,
\begin{subequations}
\begin{align}
	& \Hess_{J_{P},U}(\delta U)\nonumber\\
	 & = -\Tr( \AAdag W^{\dag}\delta U)W \AAdag + \Tr(\delta U^{\dag}W \AAdag )U \AAdag W^{\dag}U\\
	 & = -\frac{1}{2}\Big\{\langle W \AAdag - U \AAdag W^{\dag}U, \delta U\rangle(W \AAdag - U \AAdag W^{\dag}U)\nonumber\\
	 & \quad \quad+ \langle iW \AAdag + iU \AAdag W^{\dag}U, \delta U\rangle(iW \AAdag + iU \AAdag W^{\dag}U)\Big\}
\end{align}	
\end{subequations}
so that the Hessian is rank 2 except where there exists $(\alpha, \beta)\neq 0\in\mathbb{R}^{2}$ such that $\alpha\big(W \AAdag - U \AAdag W^{\dag}U\big) = \beta\big(iW \AAdag + iU \AAdag W^{\dag}U\big)$.     This is exactly the condition just considered for the surjectivity of $\rmd F$, so that the Hessian is rank 2 if and only if $\rmd F$ is surjective.  So, on the open dense set of $A$'s for which this happens, the maximum set is a nondegenerate (in the Morse-Bott sense), codimension 2 submanifold of $\UH$ (this is similar to a classical result on such distance functions \cite[Thm. 6.6]{Milnor1973}).  Since the other critical points also comprise nondegenerate submanifolds, we conclude that for these $A$'s, $J_{P}$ is a Morse-Bott function.   

Let $P := U^{\dag}W \AAdag - \AAdag W^{\dag}U\in\uH$ and $Q:=iU^{\dag}W \AAdag + i \AAdag W^{\dag}U\in\uH$.   Then letting $Y = U^{\dag}\delta U$, the Hessian eigenvalue equation at the global maximum may be written
\begin{equation}
	\lambda Y = -\frac{1}{2}\big(\langle P,Y\rangle P + \langle Q,Y \rangle Q\big).
	\label{eqn:JPmaxSkewHermHessian}
\end{equation}
Let $\gamma_{\pm}$ and $v^{\pm}$ be the eigenvalues and eigenvectors of the real, symmetric, positive semidefinite Gram matrix
\begin{equation}
	G = \begin{bmatrix}\langle P,P\rangle & \langle P,Q\rangle\\\langle Q,P\rangle & \langle Q,Q\rangle\end{bmatrix}.
\end{equation}
Then the eigenvalues of \eqref{eqn:JPmaxSkewHermHessian} are $\lambda_{\pm} = -\gamma_{\pm}/2$ and the eigenvectors are $Y^{\pm} = v_{1}^{\pm}P + v_{2}^{\pm}Q$.

\subsection{Examples}
\begin{example}[$ \AAdag $ is a projection]
Suppose that $ \AAdag $ is a projection, or equivalently, $A$ is a partial isometry.  Then $\mathcal{H} = \ES_{0}\oplus \ES_{1}$, $\omega_{0}^{2} = 0$, and $\omega_{1}^{2} = 1$.  By Theorem \ref{thm:phaseInvariantSummary}, aside from the global maximum set, $J_{P}$ admits exactly $\lceil n_{1}/2\rceil$ critical submanifolds $C_{\nu_{1}}$ for $\nu_{1}=0,\dots,\lceil n_{1}/2\rceil-1$, each isomorphic to $\mathrm{U}(\ES_{0})\oplus \big[\mathrm{U}(1)\times\GrES{\nu_{1}}{1}\big]$ and having dimension $\mathcal{N}_{0} = 1 + n_{0}^{2} + 2\nu_{1}(n_{1}-\nu_{1})$.  The critical submanifold $C_{\nu_{1}}$ has critical value $J_{P} = 4\nu_{1}(n_{1}-\nu_{1})$, and the Hessian operator at a critical point $U\in C_{\nu_{1}}$ can have (depending on $\nu_{1}$ and $n_{0}$) eigenvalues $-2n_{1}$, $-2n_{1}+4\nu_{1}$, $-n_{1}+2\nu_{1}$, $0$, $n_{1}-2\nu_{1}$, and $2n_{1}-4\nu_{1}$.  The negative Hessian eigenspace at $U$ is
\begin{align}
	E_{-} & = E_{-2n_{1}}\oplus E_{-2n_{1}+4\nu_{1}}\oplus E_{-n_{1}+2\nu_{1}}\nonumber\\
	&  = U\left[\mathbb{R}Y^{(-2n_{1})}\oplus \mathrm{su}(\ES_{1}^{-}) \oplus \frac{\mathrm{u}(\ES_{0}\oplus \ES_{1}^{-})}{\mathrm{u}(\ES_{0})\oplus \mathrm{u}(\ES_{1}^{-})}\right]
\end{align}
 of dimension $\mathcal{N}_{-}=(2n_{0}+\nu_{1})\nu_{1}$, where 
\begin{equation}
	Y^{(-2n_{1})} = \begin{cases}\frac{i}{\nu_{1}}\identity_{\ES_{1}^{-}}\oplus\frac{i}{n_{1}-\nu_{1}}\identity_{\ES_{1}^{+}} & \nu_{1}>0\\0 & \text{otherwise}.\end{cases}
\end{equation}
The positive Hessian eigenspace at $U$ is
\begin{align}
	E_{+} & = E_{2n_{1}-4\nu_{1}}\oplus E_{n_{1}-2\nu_{1}}\nonumber\\
	&  = U\left[\mathrm{su}(\ES_{1}^{+})\oplus\frac{\mathrm{u}(\ES_{0}\oplus \ES_{1}^{+})}{\mathrm{u}(\ES_{0})\oplus\mathrm{u}(\ES_{1}^{+})}\right]
\end{align}
 of dimension $\mathcal{N}_{+} = (2n_{0}+n_{1}-\nu_{1})(n_{1}-\nu_{1}) - 1$.  The full set of critical values for this problem is $\{0,4(n_{1}-1),8(n_{1}-2),\dots, 4\lfloor (n_{1}^{2}-1)/4\rfloor, n_{1}^{2}\}$, where the final value $n_{1}^{2} = \|A\|^{4}$ is the globally maximal value.  In the particular case that $ \AAdag $ is fully degenerate (e.g., $A = \identity$), it is found that the maximum set of $J_{P}$ is a nondegenerate submanifold if and only if $N=\dim(\mathcal{H})$ is odd.  However, when $N$ is even, arbitarily small perturbations of $A$ about $\identity$ are sufficient to obtain a Morse-Bott function.
\end{example}

\begin{example}[$\omega_{i}^{2} = 2^{i}$]
Suppose that $ \AAdag $ is nondegenerate and nonsingular, so that $\mathcal{H} = \bigoplus_{i=1}^{N} \ES_{i}$, $n_{i}=1$ for all $i=1,\dots, N$, and $\omega_{1}^{2} < \dots < \omega_{N}^{2}$.  Suppose further that $\omega_{i}^{2} = 2^{i}$ for $i=1,\dots,N$.  Then $\omega_{N}^{2}>\sum_{i=1}^{N-1}\omega_{i}^{2}$, so $\sum_{i=1}^{N}\omega_{i}^{2}(n_{i}-2\nu_{i}) > 0$ if and only if $\nu_{N} = 0$.  Therefore, by Theorem \ref{thm:phaseInvariantSummary}, the critical set of $J_{P}$ comprises 1 global maximal set and exactly $2^{N-1}$ critical submanifolds $C_{\{\nu_{i}\}}$, each isomorphic to $U(1)$ and equal to a global phase orbit $\{\phase  U\;:\; \phase \in\mathrm{U}(1)\}$ of a single unitary operator.  For the critical submanifold identified by indices $(\nu_{1},\dots,\nu_{N})$, the critical value is 
\begin{align}
	J_{P} & = 4\left(\sum_{i=1}^{N-1}2^{i}\nu_{i}\right)\left(\sum_{i=1}^{N}2^{i}(n_{i}-\nu_{i})\right)\nonumber\\
	& = 4\left(\sum_{i=1}^{N-1}2^{i}\nu_{i}\right)\left[2^{N+1}-2 - \sum_{i=1}^{N-1}2^{i}\nu_{i}\right],
\end{align}
the negative Hessian eigenbundle has dimension $\mathcal{N}_{-} = \sum_{i=1}^{N-1}(2i-1)\nu_{i}$, and the positive eigenbundle has dimension $\mathcal{N}_{+} = N^{2} - 1 - \sum_{i=1}^{N-1}(2i-1)\nu_{i}$.  The full set of critical values for this problem is $\{16(2^{N}-2),32(2^{N}-3),48(2^{N}-4),\dots, 2^{N+1}(2^{N+1}-4), (2^{N+1}-2)^{2}\}$.  The last of these critical values corresponds to the global maximum.
\end{example}

\bigskip

\section{Landscapes Based on Intrinsic Distance}\label{sec:intrinsicDist}
The kinematic landscapes $J_{F}$ and $J_{P}$ considered above are based on the Euclidean (or norm) distance on $\UH$ and $\PUH$, respectively.  We now describe two additional distance measures based on the intrinsic distance between operators in $\UH$ and $\PUH$ under the Riemannian metric induced by the real Hilbert-Schmidt inner product on $\BH$.  

The first of these distance measures is quite simple to define.  Since the chosen Riemannian metric is bi-invariant on $\UH$, any geodesic starting at $U\in\UH$ is of the form $\gamma(s) = Ue^{As}$ for some $A\in\uH$.  To find a geodesic joining $U$ to some target $W\in\UH$, let $Ue^{A} = \gamma(1) = W$, so that $e^{A} = U^{\dag}W$ and $A = \log(U^{\dag}W)$.  This matrix logarithm is not uniquely defined, but the length of the geodesic $\gamma$ defined on the interval $[0,1]$ is given by $L[\gamma] = \int_{0}^{1}\|\gamma'(s)\|\,ds = \|A\|$.  The minimum such length is obtained by taking $A=\log(U^{\dag}W)$ from the principal branch of the logarithm so that all eigenvalues lie in $(-i\pi, i\pi]$.  We then define the landscape as 
\begin{equation}
	J_{G}(U) := \frac{1}{2}\|\log(U^{\dag}W)\|^{2}. 
\end{equation} Then the gradient of $J_{G}$ is given by (see Appendix \ref{sec:gradJG})
\begin{equation}
	\grad J_{G}(U) = -U\log(U^{\dag}W).
\end{equation}
As most numerical matrix logarithm routines (e.g., the {\tt {logm}} function in MATLAB) compute the principal branch, they provide a ready means to obtain both the landscape value and the gradient.  Since the norm of $\grad J_{G}$ is the distance to the target, this vector field is only zero at the target, i.e. the global minimum of the landscape.  Hence, there are no traps or saddles.  The gradient field has the property that it is discontinuous and multiply defined at the cut loci of $\UH$ (where the spectrum of $U^{\dag}W$ contains $-1$), but this is not a problem for an optimal control algorithm since the matrix logarithm routine will have to choose one from among the possible solutions, all of which describe minimal geodesics to the target that are equally satisfactory.

A phase-invariant version of $J_{G}$ may be constructed analogously by considering minimal geodesics on the projective unitary group $\PUH\simeq \UH/\mathrm{U}(1)$, or equivalently by defining $J_{GP}(U):= \min_{\phase\in\mathrm{U}(1)}\frac{1}{2}\|\log(\phase U^{\dag}W)\|^{2}$ on $\UH$.  It may be shown that 
\begin{subequations}
\begin{align}
	J_{GP} & = \min_{k\in\mathbb{Z}_{N}}\frac{1}{2}\left\|\log\left(e^{\frac{2\pi i k}{N}}\det(U^{\dag}W)^{-\frac{1}{N}}U^{\dag}W\right)\right\|^{2}\label{eqn:phaseInvGeod}\end{align}  
\begin{align}
	&\grad J_{GP}(U)\nonumber\\
	& = -U\bigg\{\log\left(e^{\frac{2\pi i k}{N}}\det(U^{\dag}W)^{-\frac{1}{N}}U^{\dag}W\right)\nonumber\\
	& \qquad  - \Tr\left[\log\left(e^{\frac{2\pi i k}{N}}\det(U^{\dag}W)^{-\frac{1}{N}}U^{\dag}W\right)\right]\frac{\identity}{N}\bigg\},\label{eqn:phaseInvGeodGrad}
\end{align}  
\end{subequations}
where $k$ in \eqref{eqn:phaseInvGeodGrad} is the minimizer from \eqref{eqn:phaseInvGeod}.  With this minimizing $k$, the trace in \eqref{eqn:phaseInvGeodGrad} will be zero, so that 
\begin{equation}
	\grad J_{GP}(U) = -U\log\left(e^{\frac{2\pi i k}{N}}\det(U^{\dag}W)^{-\frac{1}{N}}U^{\dag}W\right).
\end{equation}
As with $J_{G}$, the norm of $\grad J_{GP}$ is the distance to the target, and this vector field is only zero at the target, i.e. the global minimum of the landscape.  Hence, there are no traps or saddles.  One downside to this landscape is that it appears that all $N$ possible values of $k$ must be tried in order to find the minimizer of \eqref{eqn:phaseInvGeod}.  This behavior has a topological interpretation on $\PUH$.  Since the fundamental group of $\PUH$ is $\pi_{1}(\PUH) \cong \mathbb{Z}_{N} = \mathbb{Z}/N\mathbb{Z}$, there are exactly $N$ homotopy classes of paths connecting $\pi(U)$ to the target $\pi(W)$.  Within each of these classes is a unique minimal geodesic, and these $N$ minimal geodesics are identified by the vectors 
\begin{align}
	U\bigg\{\log&\left(e^{\frac{2\pi i k}{N}}\det(U^{\dag}W)^{-\frac{1}{N}}U^{\dag}W\right) \nonumber\\
	& \mbox{} - \frac{1}{N}\Tr\left[\log\left(e^{\frac{2\pi i k}{N}}\det(U^{\dag}W)^{-\frac{1}{N}}U^{\dag}W\right)\right]\identity\bigg\}
\end{align}
indexed by $k$.

A distance metric based on intrinsic distance could in principle be applied to the case where only some of the states are important, analogous to $J_{F}$ and $J_{P}$ where $A$ is rank deficient (e.g. where A is a projector).  This is equivalent to computing the geodesic distance between points on the Stiefel manifold $V_{\mathcal{H}/\ES_{0}}(\mathcal{H})\simeq \UH/\big(\identity_{\ES_{0}}\oplus\mathrm{U}(\mathcal{H}/\ES_{0})\big)$ or on its projective cousin $V_{\mathcal{H}/\ES_{0}}(\mathcal{H})/\mathrm{U}(1)$.  However, the two-point geodesics on these spaces are non-trivial to compute.  The calculation requires solution of a boundary value problem or an optimization problem to find each minimal geodesic.  For that reason, these intrinsic distance metrics may not be practical for this scenario.

\bigskip

\section{Dynamical Critical Point Analysis}\label{sec:dynamicalAnalysis}
Now that we have elucidated the structure of the critical sets of the kinematic landscapes $J_{F}$ and $J_{P}$, we return to the problem of characterizing the critical set of the dynamical landscapes $\tilde{J} = J\circ U_{T}$.  Let $\mathcal{M}\subset \UH$ be one of the critical submanifolds identified in the previous sections.  It can be proved (see Appendix \ref{app:differentiability}) that $U_{T}:\mathbb{K}\rightarrow \UH$ is $C^{\infty}$ (i.e., infinitely Fr\'echet differentiable).  In addition, since $\UH$ is finite-dimensional, if $U_{T}(\mathcal{E})\in \mathcal{M}$ then $(\rmd_{\mathcal{E}}U_{T})^{-1}\big(\rmT_{U_{T}(\mathcal{E})}\mathcal{M}\big)$ has finite codimension, so is closed and has a closed complement (i.e., it ``splits'').  Therefore, away from singular points of $U_{T}$ (i.e., those $\mathcal{E}\in\mathbb{K}$ such that $\rmd_{\mathcal{E}}U_{T}$ is rank-deficient), $U_{T}$ is \emph{transversal} to $\mathcal{M}$ and by the transversal mapping theorem \cite{Abraham1988}, $U_{T}^{-1}(\mathcal{M})$ is a Hilbert submanifold of $\mathbb{K}$, $\rmT_{\mathcal{E}}\big(U_{T}^{-1}(\mathcal{M})\big) = (\rmd_{\mathcal{E}}U_{T})^{-1}\big(\rmT_{U_{T}(\mathcal{E})}\mathcal{M}\big)$, and $\codim\big(U_{T}^{-1}(\mathcal{M})\big) = \codim(\mathcal{M})$.

Let $\mathcal{E}\in\mathbb{K}$ be a regular critical point of $\tilde{J}$, i.e. such that $\grad\tilde{J}(\mathcal{E}) = 0$ and $\rmd_{\mathcal{E}}U_{T}$ is full rank.  It may be seen that at such a point, the Hessian of $\tilde{J}$ is given by $\Hess_{\tilde{J},\mathcal{E}} = (\rmd_{\mathcal{E}} U_{T})^{*}\circ\Hess_{J,U_{T}(\mathcal{E})}\circ (\rmd_{\mathcal{E}} U_{T})$.  Let $\mathcal{A}_{\mathcal{E}}$ be the linear operator on $\rmT_{U_{T}(\mathcal{E})}\UH$ given by 
\begin{equation}
	\mathcal{A}_{\mathcal{E}} = \big(\rmd_{\mathcal{E}}U_{T}\circ (\rmd_{\mathcal{E}}U_{T})^{*}\big)^{\frac{1}{2}}\circ\Hess_{J,U_{T}(\mathcal{E})}\circ\big(\rmd_{\mathcal{E}}U_{T}\circ (\rmd_{\mathcal{E}}U_{T})^{*}\big)^{\frac{1}{2}}.
\end{equation}
Since $\rmd_{\mathcal{E}} U_{T}$ is assumed to have full rank, we may invoke Sylvester's law of inertia \cite{Horn1985} to conclude that $\mathcal{A}_{\mathcal{E}}$ and $\Hess_{J,U_{T}(\mathcal{E})}$ have the same numbers of positive, negative, and zero eigenvalues.  Let $\{(\eta_{j}, Q_{j})\}$ for $j=1,\dots,N^{2}$ be the eigenvalues and eigenvectors of $\mathcal{A}_{\mathcal{E}}$, and let $Z_{j} = (\rmd_{\mathcal{E}}U_{T})^{*}\circ\big(\rmd_{\mathcal{E}}U_{T}\circ (\rmd_{\mathcal{E}}U_{T})^{*}\big)^{-\frac{1}{2}}Q_{j}$.  Then 
\begin{align}
	\Hess_{\tilde{J},\mathcal{E}}Z_{j} & = (\rmd_{\mathcal{E}} U_{T})^{*}\!\circ\!\Hess_{J,U_{T}(\mathcal{E})}\circ \big(\rmd_{\mathcal{E}}U_{T}\!\circ\! (\rmd_{\mathcal{E}}U_{T})^{*}\big)^{\frac{1}{2}}Q_{j} \nonumber\\
	& = \eta_{j}(\rmd_{\mathcal{E}}U_{T})^{*}\circ\big(\rmd_{\mathcal{E}}U_{T}\circ (\rmd_{\mathcal{E}}U_{T})^{*}\big)^{-\frac{1}{2}}Q_{j}\nonumber\\
	& = \eta_{j}Z_{j}
\end{align}
so that $\{(\eta_{j},Z_{j})\}$ for $j=1,\dots,N^{2}$ are eigenvalues and eigenvectors of $\Hess_{\tilde{J},\mathcal{E}}$.  Because $\Hess_{\tilde{J},\mathcal{E}}$ is self-adjoint, any other eigenvector $Z$ must be orthogonal to the $\{Z_{j}\}$.  Also, note that since the $\{Q_{j}\}$ span $\rmT_{U_{T}(\mathcal{E})}\UH$, the $\{Z_{j}\}$ span $\Range \big((\rmd_{\mathcal{E}}U_{T})^{*}\big)$.  Then, for any $X\in\rmT_{U_{T}(\mathcal{E})}\UH$, $ 0 = \langle Z, (\rmd_{\mathcal{E}}U_{T})^{*}(X)\rangle = \langle \rmd_{\mathcal{E}}U_{T}(Z), X\rangle$, so that $\rmd_{\mathcal{E}}U_{T}(Z) = 0$ and therefore $\Hess_{\tilde{J},\mathcal{E}}Z = 0$.  Thus, $\Hess_{\tilde{J},\mathcal{E}}$ has infinitely many eigenvalues; $N^{2}$ of them are identical to the eigenvalues of $\mathcal{A}_{\mathcal{E}}$, and the remaining infinite number of eigenvalues are all zero.  Since $J$ has no local traps, we can conclude that $\tilde{J}$ has no local traps among the regular critical points.  From the transversal mapping theorem we find that $\rmT_{\mathcal{E}}\big(U_{T}^{-1}(\mathcal{M})\big) = (\rmd_{\mathcal{E}}U_{T})^{-1}\big(\rmT_{U_{T}(\mathcal{E})}\mathcal{M}\big)$, implying that for any $f\in\mathbb{K}$, we have 
\begin{subequations}
\begin{align}
	 \Hess_{\tilde{J},\mathcal{E}}(f) &= (\rmd_{\mathcal{E}} U_{T})^{*}\circ\Hess_{J,U_{T}(\mathcal{E})}\circ (\rmd_{\mathcal{E}} U_{T})(f) = 0\nonumber\\
	& \Longleftrightarrow \quad \rmd_{\mathcal{E}} U_{T}(f)\in \ker\Hess_{J,U_{T}(\mathcal{E})}\\
	& \Longleftrightarrow \quad \rmd_{\mathcal{E}} U_{T}(f)\in \rmT_{U_{T}(\mathcal{E})}\mathcal{M}\\
	& \Longleftrightarrow \quad f\in(\rmd_{\mathcal{E}}U_{T})^{-1}\big(\rmT_{U_{T}(\mathcal{E})}\mathcal{M}\big)\\
	& \Longleftrightarrow \quad f\in \rmT_{\mathcal{E}}\big(U_{T}^{-1}(\mathcal{M})\big).
\end{align}
\end{subequations}
Hence, the null space of $\Hess_{\tilde{J},\mathcal{E}}$ is identical to $\rmT_{\mathcal{E}}\big(U_{T}^{-1}(\mathcal{M})\big)$, the tangent space to the critical submanifold.  

In the case where the Hamiltonian takes the dipole form $H(t) = H_{0} - \mathcal{E}(t)\mu$ for any $\mathcal{E}\in L^{2}(\mathbb{R}_{+};\mathbb{R})$, the Fr\'echet derivative of $U_{T}$ is given by \cite{Dominy2008}
\begin{equation}
	\rmd_{\mathcal{E}}U_{T}(\delta \mathcal{E}) = \frac{i}{\hbar}U_{T}(\mathcal{E})\int_{0}^{T}U_{t}^{\dag}(\mathcal{E})\mu U_{t}(\mathcal{E})\delta\mathcal{E}(t)\rmd t.
	\end{equation}
Then the adjoint operator of the derivative is 
\begin{equation}
	\rmd_{\mathcal{E}}U_{T}^{*}(A)(t) = -\Im\Tr\big(A^{\dag}U_{T}(\mathcal{E})U_{t}^{\dag}(\mathcal{E})\mu U_{t}(\mathcal{E})\big)
\end{equation}
 for any $A\in \rmT_{U_{T}(\mathcal{E})}\UH$, and the operator norm of this adjoint is uniformly bounded by $\|\rmd_{\mathcal{E}}U_{T}^{*}\| \leq \sqrt{T}\|\mu\|$.  For any smooth ``kinematic'' function $g:\UH\rightarrow \mathbb{R}$, let $\tilde{g} = g\circ U_{T}$ be the corresponding ``dynamical'' function on $L^{2}(\mathbb{R}_{+};\mathbb{R})$.  Then $\grad\tilde{g}(\mathcal{E}) = \rmd_{\mathcal{E}}U_{T}^{*}(\grad g(U_{T}(\mathcal{E})))$ and 
\begin{subequations}
\begin{align}
 	\|\grad\tilde{g}(\mathcal{E})\| & \leq \|\rmd_{\mathcal{E}}U_{T}^{*}\|\|\grad g(U_{T}(\mathcal{E}))\|\\
	& \leq \sqrt{T}\|\mu\|\|\grad g(U_{T}(\mathcal{E}))\|.
\end{align}
\end{subequations}
Since $g$ is smooth, $\|\grad g\|$ is continuous over $\UH$, so that since $\UH$ is compact, $\|\grad g\|$ is uniformly bounded.  Therefore, $\|\grad \tilde{g}\|$ is uniformly bounded over $L^{2}(\mathbb{R}_{+};\mathbb{R})$.  For any dynamical quantum control landscape constructed in this way, in particular the landscapes considered in the present paper, the slope of the landscape (i.e. the speed of the gradient flow) is uniformly bounded by some constant.  

Taken together, these results show that, even though the control space $\mathbb{K}$ is unbounded and infinite-dimensional and one might naively expect anything to happen, the landscapes under consideration are well-behaved, exhibiting gradient flows which do not get trapped (at least away from singular points) and which do not speed out of control.

\bigskip

\section{Summary}\label{sec:summary}
This work presented an expanded analysis of landscapes $J_{F}$ and $J_{P}$, which are based on the Euclidean distances between unitary operators in $\UH$ and $\PUH$, respectively.  The expansion appears in several ways.  First, additional freedom has been allowed in the landscape functions themselves, by admitting $A$ matrices that are rank-deficient.  Landscapes based on these rank-deficient $A$ matrices measure the distance between unitary operators by their action on a subspace of the full state space.  This can be the desired objective for designing a quantum information processor, for example, where only this subspace of the state space is to be used for the quantum register.  This additional freedom in defining the landscape is consistent with the principal finding of earlier work on landscapes of this form: they have no suboptimal minima (i.e., ``traps'') that could impede a deterministic optimal control algorithm (such as gradient descent) from reaching the global minimum.

In addition to broadening the families of landscapes for consideration, we have provided more detail on the structure of the critical sets and the behavior of the landscape functions at these critical sets.  The critical sets were shown to generally be disjoint unions of critical submanifolds and we have described the structure of these submanifolds, as products of Grassmann manifolds and unitary groups.  Furthermore, we have shown that these critical submanifolds are generally nondegenerate in the Morse-Bott sense, so that the kinematic landscapes are generally Morse-Bott functions.

These results were related back to the corresponding dynamical landscapes through the control-to-propagator map $U_{T}$, implicitly defined by the Schr\"odinger equation, that takes a control function as input and returns the final time unitary evolution operator.  This map was shown to be infinitely Fr\'echet differentiable, leading to the conclusion that, away from the singular points of $U_{T}$, the level sets and critical sets of the dynamical landscapes are $C^{\infty}$ smooth, finite codimension submanfolds of the infinite-dimensional control space $\mathbb{K} = L^{2}(\mathbb{R}_{+};\mathbb{R})$.  Also, the number of positive and negative Hessian eigenvalues (and therefore the characterization as a minimum, maximum, or saddle) was shown to be identical for a kinematic critical point and a regular point of $U_{T}$ that maps to it.  This behavior implies that no traps exist in the dynamical landscape among the set of regular points of $U_{T}$.  Furthermore, Morse-Bott nondegeneracy of the critical set is also preserved away from singular points of $U_{T}$, which can be important for certain numerical landscape exploration methods such as second order D-MORPH \cite{Beltrani2011}. 

Finally, two additional landscapes were introduced that are based on the intrinsic or geodesic distance between operators in $\UH$ and $\PUH$, respectively, rather than Euclidean distance.  These kinematic landscapes have the desirable property of having no critical points except for the global minimum at the target.  These landscapes may allow for more efficient performance of optimal control algorithms over $J_{F}$ and $J_{P}$, since the latter have many saddle points where the gradient is zero.

\bigskip

\appendices

\section{Differentiability of \texorpdfstring{$U(T,0)$}{U(T,0)} With Respect to the Control}
\label{app:differentiability}
Let $\MH\subset \BH$ denote the space of Hermitian operators endowed with the real Hilbert-Schmidt inner product $\langle A,B\rangle_{\mathrm{HS}} = \Re\Tr(A^{\dag}B)$, and let $\mathbb{H}(\mathcal{H}) = L^{2}\big(\mathbb{R}_{+};\MH\big) = \MH\otimes L^{2}(\mathbb{R}_{+}; \mathbb{R})$ denote the space of all square-integrable time-dependent Hamiltonians on $\mathcal{H}$ with inner product 
\begin{equation}
	\langle H_{1},H_{2}\rangle_{L^{2}} = \int_{0}^{\infty}\Re\Tr[H_{1}(t)H_{2}(t)]\,dt.
\end{equation}
Let $Z_{T}:\mathbb{H}(\mathcal{H})\rightarrow \UH$ be the map, defined implicitly through the Schr\"odinger equation, that takes a time-dependent Hamiltonian $H(\cdot)\in\Herm $ and produces the corresponding unitary time-evolution operator at time $T$: $U(T,0)\in\UH$.  This map is well-defined over the entire domain because of the absolute convergence of the Dyson series over $\Herm $:
\begin{align}
	Z_{T}(H) & = \identity +\left(-\frac{i}{\hbar}\right)\int_{0}^{T}\rmd t_{1}\,H(t_{1})\nonumber\\
	& \quad + \left(-\frac{i}{\hbar}\right)^{2}\int_{0}^{T}\rmd t_{1}\,H(t_{1})\int_{0}^{t_{1}}\rmd t_{2}\,H(t_{2}) + \dots.
\end{align}  
In this appendix, we will prove that $Z_{T}$ is infinitely Fr\'echet differentiable over $\Herm $.  A corollary is that the map $U_{T}:\mathbb{K}\to\UH$ defined in the body of the paper is infinitely Fr\'echet differentiable over all of $\mathbb{K} = L^{2}(\mathbb{R}_{+};\mathbb{R})$.

\begin{lemma}
	If $f:[a,b]\rightarrow \mathbb{R}$ is integrable, then
	\begin{align}
		\int_{a}^{b}\rmd t_{1}\,f(t_{1})&\int_{a}^{t_{1}}\rmd t_{2}\,f(t_{2})\cdots\int_{a}^{t_{n-1}}\rmd t_{n}\,f(t_{n})\nonumber\\
		 & = \frac{1}{n!}\left(\int_{a}^{b}f(t)\,\rmd t\right)^{n}.\label{eqn:DysonFactorial}
	\end{align}
	\label{lem:DysonFactorial}
\end{lemma}
\begin{IEEEproof}
	Note first that \eqref{eqn:DysonFactorial} holds trivially for $n=1$.  Suppose that it holds for $n=m$.  Then
	\begin{subequations}
	\begin{align}
		\int_{a}^{b}&\rmd t_{1}\,f(t_{1})\int_{a}^{t_{1}}\rmd t_{2}\,f(t_{2})\cdots\int_{a}^{t_{m}}\rmd t_{m+1}\,f(t_{m+1})\nonumber\\
		& = \frac{1}{m!}\int_{a}^{b}f(t_{1})\left(\int_{a}^{t_{1}}f(t)\,\rmd t\right)^{m}\,\rmd t_{1}\\
		& = \frac{1}{m!}\int_{a}^{b}\frac{\rmd}{\rmd t_{1}}\left[\frac{1}{m+1}\left(\int_{a}^{t_{1}}f(t)\,\rmd t\right)^{m+1}\right]\,\rmd t_{1}\\
		& = \frac{1}{(m+1)!}	\left(\int_{a}^{b}f(t)\,\rmd t\right)^{m+1}
	\end{align}
	\end{subequations}
	and the lemma follows for arbitrary $n\in\mathbb{N}$ by induction.
\end{IEEEproof}

\begin{definition}
For integrable operator-valued functions $A_{i}:[0,T]\rightarrow \BH$ and for integrable real-valued functions $a_{i}:[0,T]\rightarrow \mathbb{R}$, we will use the following short-hand notation for the Dyson-esque terms
\begin{subequations}
\begin{align}
	&\Upsilon_{T}[A_{1},A_{2},\dots, A_{n}]\nonumber\\
	& := \int_{0}^{T}\rmd t_{1}\,A_{1}(t_{1})\int_{0}^{t_{1}}\rmd t_{2}\,A_{2}(t_{2})\cdots\int_{0}^{t_{n-1}}\rmd t_{n}\,A_{n}(t_{n})\\
	&\Upsilon_{T}[a_{1},a_{2},\dots, a_{n}]\nonumber\\
	& := \int_{0}^{T}\rmd t_{1}\,a_{1}(t_{1})\int_{0}^{t_{1}}\rmd t_{2}\,a_{2}(t_{2})\cdots\int_{0}^{t_{n-1}}\rmd t_{n}\,a_{n}(t_{n}).
\end{align}
\end{subequations}
\end{definition}

\begin{lemma}
	If $g_{1}, \dots, g_{m}\in L^{2}(\mathbb{R}_{+};\mathbb{R}_{+})$ are non-negative square-integrable functions on $[0,\infty)$, then 
	\begin{equation}\Upsilon_{T}[g_{1},g_{2},\cdots g_{m}]\leq T^{m/2}\|g_{1}\|_{L^{2}}\cdots\|g_{m}\|_{L^{2}}\end{equation}
	\label{lem:nodupDysonBound}
\end{lemma}
\begin{IEEEproof}
	Since the $g_{i}$'s are non-negative functions, we get the inequalities
	\begin{subequations}
	\begin{align}
		&\Upsilon_{T}[g_{1},g_{2},\cdots g_{m}]\nonumber\\
		& \leq \int_{0}^{T}\rmd t_{1}\,g_{1}(t_{1})\int_{0}^{T}\rmd t_{2}\,g_{2}(t_{2})\cdots\int_{0}^{T}\rmd t_{m}\,g_{m}(t_{m})\\
		& \leq T^{m/2}\prod_{i=1}^{m}\left(\int_{0}^{T}\rmd t_{i}\,g_{i}^{2}(t_{i})\right)^{\frac{1}{2}}\\
		& \leq T^{m/2}\|g_{1}\|_{L^{2}}\cdots\|g_{m}\|_{L^{2}}
	\end{align}
	\end{subequations}
	by extension of the integrals out to the interval $[0,T]$, followed by application of the Cauchy-Schwarz inequality, and finally extension out to $[0,\infty)$.
\end{IEEEproof}

\begin{lemma}
	If $f, g_{1}, \dots, g_{m}\in L^{2}(\mathbb{R}_{+};\mathbb{R}_{+})$ are non-negative square-integrable functions on $[0,\infty)$, then 
	\begin{align}
		&\Upsilon_{T}[\underbrace{f,\dots, f}_{\beta_{0} \text{ terms}},g_{1},\underbrace{f,\dots, f}_{\beta_{1} \text{ terms}},g_{2},f,\dots,f,g_{m},\underbrace{f,\dots, f}_{\beta_{m} \text{ terms}}] \nonumber\\
		&\qquad  \leq \frac{T^{\frac{1}{2}\sum\beta_{i}}}{\prod_{i=0}^{m}\beta_{i}!}\|f\|_{L^{2}}^{\sum\beta_{i}}\Upsilon_{T}[g_{1},g_{2},\cdots g_{m}]	
	\end{align}
\label{lem:dupsDysonBound}
\end{lemma}
\begin{IEEEproof}
Let $\sigma_{i} = \beta_{0}+\dots +\beta_{i} + i +1$ for $i=0,\dots,m$.  Then, using Fubini's theorem, we may rearrange the order of integration as follows:
	\begin{subequations}
	\begin{align}
		\lefteqn{\Upsilon_{T}[\underbrace{f,\dots, f}_{\beta_{0} \text{ terms}},g_{1},\underbrace{f,\dots, f}_{\beta_{1} \text{ terms}},g_{2},f,\dots,f,g_{m},\underbrace{f,\dots, f}_{\beta_{m} \text{ terms}}]}\nonumber\\
		 & = \int_{0}^{T}\rmd t_{\sigma_{0}}\,g_{1}(t_{\sigma_{0}})\int_{0}^{t_{\sigma_{0}}}\rmd t_{\sigma_{1}}\,g_{2}(t_{\sigma_{1}})\cdots\nonumber\\
		 & \qquad\int_{0}^{t_{\sigma_{m-2}}}\hspace{-10pt}\rmd t_{\sigma_{m-1}}\,g_{m}(t_{\sigma_{m-1}})\int_{t_{\sigma_{0}}}^{T}\rmd t_{1}\,f(t_{1})\int_{t_{\sigma_{0}}}^{t_{1}}\rmd t_{2}\,f(t_{2})\nonumber\\
		 & \quad\cdots\int_{t_{\sigma_{0}}}^{t_{\sigma_{0}-2}}\rmd t_{\sigma_{0}-1}\,f(t_{\sigma_{0}-1})\int_{t_{\sigma_{1}}}^{t_{\sigma_{0}}}\rmd t_{\sigma_{0}+1}\,f(t_{\sigma_{0}+1})\nonumber\\
		 & \qquad\int_{t_{\sigma_{1}}}^{t_{\sigma_{0}+1}}\rmd t_{\sigma_{0}+2}\,f(t_{\sigma_{0}+2})\cdots\int_{t_{\sigma_{1}}}^{t_{\sigma_{1}-2}}\rmd t_{\sigma_{1}-1}\,f(t_{\sigma_{1}-1})\nonumber\\
		 & \quad\cdots\int_{t_{\sigma_{m-1}}}^{t_{\sigma_{m-2}}}\hspace{-16pt}\rmd t_{\sigma_{m-2}+1}\,f(t_{\sigma_{m-2}+1})\int_{t_{\sigma_{m-1}}}^{t_{\sigma_{m-2}+1}}\hspace{-26pt}\rmd t_{\sigma_{m-2}+2}\,f(t_{\sigma_{m-2}+2})\nonumber\\
		 &\quad\cdots\int_{t_{\sigma_{m-1}}}^{t_{\sigma_{m-1}-2}}\hspace{-22pt}\rmd t_{\sigma_{m-1}-1}\,f(t_{\sigma_{m-1}-1})\int_{0}^{t_{\sigma_{m-1}}}\hspace{-19pt}\rmd t_{\sigma_{m-1}+1}\,f(t_{\sigma_{m-1}+1})\nonumber\\
		 & \quad \cdots\int_{0}^{t_{\sigma_{m}-2}}\rmd t_{\sigma_{m}-1}\,f(t_{\sigma_{m}-1})\\
		 & = \frac{1}{\prod_{i=0}^{m}\beta_{i}!}\int_{0}^{T}\rmd t_{\sigma_{0}}\,g_{1}(t_{\sigma_{0}})\int_{0}^{t_{\sigma_{0}}}\rmd t_{\sigma_{1}}\,g_{2}(t_{\sigma_{1}})\cdots\nonumber\\
		 & \quad \cdots\int_{0}^{t_{\sigma_{m-2}}}\rmd t_{\sigma_{m-1}}\,g_{m}(t_{\sigma_{m-1}}) \left(\int_{t_{\sigma_{0}}}^{T}f(t)\,\rmd t\right)^{\beta_{0}}\times\nonumber\\
		 & \qquad \times\left(\int_{t_{\sigma_{1}}}^{t_{\sigma_{0}}}f(t)\,\rmd t\right)^{\beta_{1}}\cdots \left(\int_{0}^{t_{\sigma_{m-1}}} f(t)\,\rmd t\right)^{\beta_{m}}
	\end{align}
	\end{subequations}
	where the last step follows from Lemma \ref{lem:DysonFactorial}.	Then, since $f$ is a non-negative function, we get the inequality
	\begin{subequations}
	\begin{align}
		\lefteqn{\Upsilon_{T}[\underbrace{f,\dots, f}_{\beta_{0} \text{ terms}},g_{1},\underbrace{f,\dots, f}_{\beta_{1} \text{ terms}},g_{2},f,\dots,f,g_{m},\underbrace{f,\dots, f}_{\beta_{m} \text{ terms}}]}\nonumber\\
		 & \leq \frac{1}{\prod_{i=0}^{m}\beta_{i}!}\left(\int_{0}^{T}f(t)\,\rmd t\right)^{\sum_{i=0}^{m}\beta_{i}}\times\nonumber\\
		 & \qquad \times\int_{0}^{T}\rmd t_{1}\,g_{1}(t_{1})\int_{0}^{t_{1}}\rmd t_{2}\,g_{2}(t_{2})\cdots\int_{0}^{t_{m-1}}\rmd t_{m}\,g_{m}(t_{m})\\
		 & \leq \frac{T^{\frac{1}{2}\sum\beta_{i}}}{\prod_{i=0}^{m}\beta_{i}!}\big\|f\big\|_{L^{2}}^{\sum\beta_{i}}\Upsilon_{T}[g_{1},g_{2},\dots, g_{m}]
	\end{align}
	\end{subequations}
	by first extending the $f$ integrals to the interval $[0,T]$. and then invoking the Cauchy-Schwarz inequality.
\end{IEEEproof}

\begin{definition}
	For $m=0,1,2,\dots$, let $\mathcal{B}^{m}\big(\rmT\Herm ; \BH \big)$ denote the space of bounded $m$-multilinear operators from $\big(\rmT\Herm \big)^{m} = \rmT\Herm \times \rmT\Herm \times\dots\times \rmT\Herm $ to $\BH $, with the norm \begin{equation}\|A\| = \sup_{\{\|\delta H_{j}\|\neq 0\}}\frac{\|A(\delta H_{1},\dots,\delta H_{m})\|}{\|\delta H_{1}\|\cdots\|\delta H_{m}\|}\end{equation} for each $A\in\mathcal{B}^{m}\big(\rmT\Herm ; \BH \big)$.  Then let $\varphi_{T,m}:\mathbb{H}\rightarrow \mathcal{B}^{m}\big(\rmT\Herm ; \BH \big)$ be defined by 
	\begin{align}
		& \varphi_{T,m}(H)(\delta H_{1},\delta H_{2}, \cdots, \delta H_{m})\nonumber\\
		& := \sum_{n=m}^{\infty}\;\sum_{a_{0} + \dots + a_{m} = n-m}\;\sum_{\pi\in S_{m}}\left(-\frac{i}{\hbar}\right)^{n}\times\nonumber\\
		& \qquad \times \Upsilon_{T}[\underbrace{H,\dots,H}_{a_{0} \text{ terms}}, \delta H_{\pi(1)},\underbrace{H,\dots,H}_{a_{1} \text{ terms}}, \delta H_{\pi(2)}, H, \dots\nonumber\\
		& \qquad \qquad \qquad \dots, H, \delta H_{\pi(m)}, \underbrace{H,\dots,H}_{a_{m} \text{ terms}}]
	\end{align}
	where $S_{m}$ denotes the symmetric group on $m$ elements (i.e., the group of permutations of $m$ elements).  For $m,q=0,1,2,\dots$, let $\Psi_{T,m,q}:\mathbb{H}\oplus\rmT\mathbb{H}\rightarrow \mathcal{B}^{m}\big(\rmT\Herm ; \BH \big)$ be defined by
	\begin{align}
		&\Psi_{T,m,q}(H,\delta H)(\delta H_{1},\delta H_{2}, \cdots, \delta H_{m})\nonumber\\
		& := \sum_{n=m+q}^{\infty}\;\sum_{\jmdstack{a_{0} + \dots + a_{m+q}}{= n-m-q}}\;\sum_{b_{0} + \dots + b_{m} = q}\;\sum_{\pi\in S_{m}}\left(-\frac{i}{\hbar}\right)^{n}\times\nonumber\\
		& \qquad \times\Upsilon_{T}[\underbrace{H,\dots,H}_{a_{0} \text{ terms}}, A_{1},\underbrace{H,\dots,H}_{a_{1} \text{ terms}}, A_{2}, H, \dots\nonumber\\
		& \qquad \qquad \qquad \dots, H, A_{m+q}, \underbrace{H,\dots,H}_{a_{m+q} \text{ terms}}],
	\end{align}
	where 
	\begin{align}
		& \{A_{1},A_{2},\dots, A_{m+q}\}\nonumber\\
		& = \big\{\underbrace{\delta H,\dots,\delta H}_{b_{0} \text{ terms}}, \delta H_{\pi(1)},\underbrace{\delta H,\dots,\delta H}_{b_{1} \text{ terms}}, \delta H_{\pi(2)}, \delta H, \dots\nonumber\\
		& \qquad \qquad \dots, \delta H, \delta H_{\pi(m)}, \underbrace{\delta H,\dots,\delta H}_{b_{m} \text{ terms}}\big\}
	\end{align}
\end{definition}

\begin{lemma}
	$\varphi_{T,m}$ and $\Psi_{T,m,q}$ are well-defined since their defining sums converge absolutely, and for each $H\in\mathbb{H}$ and $\delta H\in \rmT \mathbb{H}$, $\varphi_{T,m}(H)$ and $\Psi_{T,m,q}(H)(\delta H)$ are bounded $m$-multilinear operators.
\end{lemma}
\begin{IEEEproof}
	Let $f(t) = \|H(t)\|$, $g_{i}(t) = \|A_{i}(t)\|$, $h_{j}(t) = \|\delta H_{j}(t)\|$, and $h(t) = \|\delta H(t)\|$.  Then
	\begin{subequations}
	\begin{align}
		\lefteqn{\|\Upsilon_{T}[\underbrace{H,\dots,H}_{a_{0} \text{ terms}}, A_{1}, H, \dots, H, A_{m+q}, \underbrace{H,\dots,H}_{a_{m+q} \text{ terms}}]\|}\nonumber\\
		& \leq \Upsilon_{T}[\underbrace{f,\dots,f}_{a_{0} \text{ terms}}, g_{1},\underbrace{f,\dots,f}_{a_{1} \text{ terms}}, g_{2}, f, \dots, f, g_{m+q}, \underbrace{f,\dots,f}_{a_{m+q} \text{ terms}}]\\
		& \leq \frac{T^{\frac{n-m-q}{2}}}{\prod_{i=0}^{m+q} a_{i}!}\|f\|_{L^{2}}^{n-m-q}\Upsilon_{T}[g_{1},g_{2},\dots,g_{m+q}]\\
		& = \frac{T^{\frac{n-m-q}{2}}}{\prod_{i=0}^{m+q} a_{i}!}\|H\|_{L^{2}}^{n-m-q}\times\nonumber\\
		& \qquad \times\Upsilon_{T}[\underbrace{h,\dots,h}_{b_{0} \text{ terms}}, h_{\pi(1)},\underbrace{h,\dots,h}_{b_{1} \text{ terms}}, h_{\pi(2)}, h, \dots\nonumber\\
		&\qquad \qquad \qquad \dots, h, h_{\pi(m)}, \underbrace{h,\dots,h}_{b_{m} \text{ terms}}]\\
		& \leq \frac{T^{\frac{n-m}{2}}\|H\|_{L^{2}}^{n-m-q}\|h\|_{L^{2}}^{q}}{\prod_{i=0}^{m+q} a_{i}!\prod_{j=0}^{m} b_{j}!}\Upsilon_{T}[h_{\pi(1)}, h_{\pi(2)}, \dots,  h_{\pi(m)}]\\
		& \leq \frac{T^{\frac{n}{2}}\|H\|_{L^{2}}^{n-m-q}\|h\|_{L^{2}}^{q}}{\prod_{i=0}^{m+q} a_{i}!\prod_{j=0}^{m} b_{j}!}\|h_{1}\|_{L^{2}}\|h_{2}\|_{L^{2}}\cdots\|h_{m}\|_{L^{2}}\\
		& \leq \frac{T^{\frac{n}{2}}\|H\|_{L^{2}}^{n-m-q}\|\delta H\|_{L^{2}}^{q}}{\prod_{i=0}^{m+q} a_{i}!\prod_{j=0}^{m} b_{j}!}\|\delta H_{1}\|_{L^{2}}\|\delta H_{2}\|_{L^{2}}\cdots\|\delta H_{m}\|_{L^{2}}
	\end{align}
	\end{subequations}
	by appealing to Lemmas \ref{lem:nodupDysonBound} and \ref{lem:dupsDysonBound}.  Recall, as a special case of the multinomial theorem \cite{Abramowitz1972}, that
\begin{equation}
	\sum_{c_{1}+\dots+c_{r} = p}\frac{p!}{\prod_{i=1}^{r}c_{i}!} = r^{p}.
\end{equation}
	It follows that 
	\begin{subequations}
	\begin{align}
		\lefteqn{\|\Psi_{T,m,q}(H,\delta H)(\delta H_{1},\delta H_{2},\dots, \delta H_{m})\|}\nonumber\\
		& \leq \sum_{n=m+q}^{\infty}\;\sum_{\jmdstack{a_{0} + \dots + a_{m+q} = n-m-q}{b_{0} + \dots + b_{m} = q}}\;\sum_{\pi\in S_{m}}\frac{T^{\frac{n}{2}}}{\hbar^{n}\prod_{i=0}^{m+q} a_{i}!\prod_{j=0}^{m} b_{j}!}\times\nonumber\\
		&  \qquad \quad\times\|H\|_{L^{2}}^{n-m-q}\|\delta H\|_{L^{2}}^{q}\prod_{j=1}^{m}\|\delta H_{j}\|_{L^{2}}\\
		& = \sum_{n=m+q}^{\infty}\frac{m!T^{\frac{n}{2}}(m+q+1)^{n-m-q}(m+1)^{q}}{\hbar^{n}(n-m-q)!q!}\|H\|_{L^{2}}^{n-m-q}\times\nonumber\\
		& \qquad\qquad \times\|\delta H\|_{L^{2}}^{q}\prod_{j=1}^{m}\|\delta H_{j}\|_{L^{2}}\\
		& = \frac{m!T^{\frac{m+q}{2}}(m+1)^{q}}{\hbar^{m+q}q!}\exp\left(\frac{(m+q+1)\sqrt{T}}{\hbar}\|H\|_{L^{2}}\right)\times\nonumber\\
		& \qquad \times\|\delta H\|_{L^{2}}^{q}\prod_{j=1}^{m}\|\delta H_{j}\|_{L^{2}}\label{eqn:psiBound}
	\end{align}
	\end{subequations}
	so that the sum converges absolutely, $\Psi_{T,m,q}(H, \delta H)$ is a bounded $m$-multilinear operator for each $H\in\mathbb{H}$ and $\delta H\in\rmT\mathbb{H}$, and 
\begin{align}
	& \|\Psi_{T,m,q}(H,\delta H)\|\nonumber\\
	& \leq \frac{m!T^{\frac{m+q}{2}}(m+1)^{q}}{\hbar^{m+q}q!}\exp\!\left(\!\frac{(m+q+1)\sqrt{T}}{\hbar}\|H\|_{L^{2}}\!\!\right)\|\delta H\|_{L^{2}}^{q}\nonumber\\
	& <\infty.
\end{align}
	Since $\varphi_{T,m}(H) = \Psi_{T,m,0}(H,\delta H)$, this conclusion also holds for $\varphi_{T,m}(H)$.
\end{IEEEproof}

\begin{theorem}
$Z_{T}$ is infinitely Fr\'echet differentiable, i.e. $C^{\infty}$, everywhere on $\mathbb{H}$.
\end{theorem}
\begin{IEEEproof}
	We begin by establishing that $\varphi_{T,m}$ is Fr\'echet differentiable for each $m=0,1,2,\dots$, with derivative $\varphi_{T,m+1}$.  Observe that 
	\begin{align}
		\varphi_{T,m+1}&(H)(\delta H_{1},\dots, \delta H_{m}, \delta H)\nonumber\\
		& = \Psi_{T,m,1}(H, \delta H)(\delta H_{1},\dots,\delta H_{m})
	\end{align} 
	and that since the defining sum for $\varphi_{T,m}(H+\delta H)$ converges absolutely, it may be rearranged as 
	\begin{align}
		\varphi_{T,m}&(H+\delta H)(\delta H_{1},\dots,\delta H_{m})\nonumber\\
		& = \sum_{q=0}^{\infty}\Psi_{T,m,q}(H,\delta H)(\delta H_{1}, \dots, \delta H_{m}).
	\end{align}  
	Then by appealing to the bound of $\Psi_{T,m,q}$ in \eqref{eqn:psiBound}, we get
	\begin{subequations}
	\begin{align}
		\lefteqn{\|\varphi_{T,m}(H+\delta H) - \varphi_{T,m}(H) - \varphi_{T,m+1}(H)(\cdot,\dots,\cdot,\delta H)\|}\nonumber\\
		 & = \sup_{\{\|\delta H_{j}\|=1\}}\big\|\varphi_{T,m}(H+\delta H)(\delta H_{1},\dots,\delta H_{m}) \nonumber\\
		 & \qquad \qquad \qquad \mbox{} - \varphi_{T,m}(H)(\delta H_{1},\dots,\delta H_{m})\nonumber\\
		 & \qquad \qquad \qquad \mbox{} - \varphi_{T,m+1}(H)(\delta H_{1},\dots,\delta H_{m},\delta H)\big\|\\
		 & = \sup_{\{\|\delta H_{j}\|= 1\}}\bigg\|\sum_{q=2}^{\infty}\Psi_{T,m,q}(H,\delta H)(\delta H_{1}, \dots, \delta H_{m})\bigg\|\\
		 & \leq \sup_{\{\|\delta H_{j}\|=1\}}\sum_{q=2}^{\infty}\big\|\Psi_{T,m,q}(H,\delta H)(\delta H_{1}, \dots, \delta H_{m})\big\|\\
		 & \leq \sum_{q=2}^{\infty}\frac{m!T^{\frac{m+q}{2}}(m+1)^{q}}{\hbar^{m+q}q!}\times\nonumber\\
		 & \qquad \times\exp\left(\frac{(m+q+1)\sqrt{T}}{\hbar}\|H\|_{L^{2}}\right)\|\delta H\|_{L^{2}}^{q}\\
		 & = \frac{m!T^{\frac{m+2}{2}}(m+1)^{2}}{\hbar^{m+2}}\exp\left(\frac{(m+3)\sqrt{T}}{\hbar}\|H\|_{L^{2}}\right)\|\delta H\|_{L^{2}}^{2}\nonumber\\
		 & \qquad \times\exp\left[\frac{(m+1)\sqrt{T}}{\hbar}\|\delta H\|_{L^{2}}\exp\left(\frac{\sqrt{T}}{\hbar}\|H\|_{L^{2}}\right)\right].
	\end{align}
	\end{subequations}
	Hence, 
	\begin{align}
		\lim_{\|\delta H\|\to 0}\frac{1}{\|\delta H\|}\big\|&\varphi_{T,m}(H+\delta H) - \varphi_{T,m}(H)\nonumber\\
		& - \varphi_{T,m+1}(H)(\cdot,\dots,\cdot,\delta H)\big\| = 0
	\end{align} 
	and therefore $\varphi_{T,m}$ is Fr\'echet differentiable with derivative $\varphi_{T,m+1}$.  Since $Z_{T} = \varphi_{T,0}$, this implies that $Z_{T}$ is infinitely Fr\'echet differentiable, and that the $m$'th derivative of $Z_{T}$ is $\varphi_{T,m}$.
\end{IEEEproof}

\begin{lemma}
	Let $\hat{\mathcal{H}}:\mathbb{K}\rightarrow \Herm $ be defined by $\hat{\mathcal{H}}(\mathcal{E})(t) = H_{0}-\mu\mathcal{E}(t)$ for some fixed Hermitian operators $H_{0}$ and $\mu$ in $\MH$.  Then $\hat{\mathcal{H}}$ is infinitely Fr\'echet differentiable, i.e. $C^{\infty}$.
\end{lemma}
\begin{IEEEproof}
	Let $\zeta:\mathbb{K}\rightarrow\mathcal{B}(\rmT\mathbb{K}; \rmT\Herm )$ be defined by $\zeta(\mathcal{E})(\delta \mathcal{E})(t) = -\mu\delta\mathcal{E}(t)$.  For each $\mathcal{E}\in\mathbb{K}$,  $\zeta(\mathcal{E})$ is linear, and $\|\zeta(\mathcal{E})(\delta\mathcal{E})\|_{L^{2}} = \|\mu\|_{\mathrm{HS}}\|\delta\mathcal{E}\|_{L^{2}}$, so $\zeta(\mathcal{E})$ is bounded.  Now, 
	\begin{equation} \lim_{\|\delta\mathcal{E}\|\to 0}\frac{\|\hat{\mathcal{H}}(\mathcal{E} + \delta\mathcal{E}) - \hat{\mathcal{H}}(\mathcal{E}) - \zeta(\mathcal{E})(\delta\mathcal{E})\|}{\|\delta\mathcal{E}\|} = 0\end{equation}
	so that $\zeta$ is the Fr\'echet derivative of $\hat{\mathcal{H}}$.  Since $\zeta$ is constant (i.e., $\zeta(\mathcal{E})$ is the same linear operator regardless of which $\mathcal{E}\in\mathbb{K}$ is input), the higher Fr\'echet derivatives also exist and are all equal to zero.
\end{IEEEproof}

\begin{theorem}
	$U_{T} = Z_{T}\circ\hat{\mathcal{H}}:\mathbb{K}\rightarrow \UH$ is a composition of $C^{\infty}$ maps and therefore is itself a $C^{\infty}$ map.
\end{theorem}

\bigskip

\section{Gradient of \texorpdfstring{$J_{G}$}{J\_G}}\label{sec:gradJG}
Several steps in Section \ref{sec:intrinsicDist} require differentiation of expressions involving the matrix logarithm.  Since the expressions to be differentiated are all similar, this appendix will demonstrate the computation of the gradient of the kinematic landscape $J_{G}$, as the other variations follow along similar lines.  To this end, we fix some target $W\in\UH$ and recall
\begin{equation}
	J_{G}(U):=\frac{1}{2}\|\log(U^{\dagger}W)\|_{\mathrm{HS}}^{2}.
\end{equation}
The differential of this function may then be written
\begin{subequations}
\begin{align}
	\rmd_{U}&J_{G}(\delta U) \nonumber\\
	& = \langle \log(U^{\dagger}W), \rmd_{U^{\dagger}W}\log(-U^{\dagger}\delta U U^{\dagger}W)\rangle_{\mathrm{HS}}\\
	& = \langle -U\big[\big(\rmd_{U^{\dagger}W}\log\big)^{*}(\log(U^{\dagger}W))\big]W^{\dagger}U, \delta U\rangle,
\end{align}
\end{subequations}
where $(\rmd_{U^{\dagger}W}\log)^{*}$ is the adjoint (super-)operator.  As a result, the gradient, which is the dual vector in $\rmT_{U}\UH$ of the differential functional, is given by
\begin{equation}
\grad J_{G}(U) = -U\big[\big(\rmd_{U^{\dagger}W}\log\big)^{*}(\log(U^{\dagger}W))\big]W^{\dagger}U
\end{equation}
It follows from the inverse relationship of the operator logarithm and exponential that $\exp\circ\log = \id$, whence $\rmd_{\log Z}\exp\circ \rmd_{Z}\log = \id$ for any $Z\in\UH$, and therefore
\begin{equation}
	\big(\rmd_{Z}\log\big)^{*} \circ \big(\rmd_{\log Z}\exp\big)^{*}= \id.\label{eqn:expLogAdj}
\end{equation}

Now, it is well-known \cite{Karplus1948, Mathias1992} that 
\begin{equation}
	\rmd_{\log Z}\exp(X) = \int_{0}^{1}e^{s\log(Z)} X e^{(1-s)\log(Z)}\,\rmd s,
\end{equation}
and therefore
\begin{equation}
	\big(\rmd_{\log Z}\exp\big)^{*}(X) = \int_{0}^{1}e^{-s\log(Z)} X e^{-(1-s)\log(Z)}\,\rmd s,
\end{equation}
It follows that
\begin{align}
	\big(&\rmd_{\log U^{\dagger}W}\exp\big)^{*}\big(U^{\dagger}W\log(U^{\dagger}W)\big) \nonumber\\
	& = \int_{0}^{1}e^{-s\log(U^{\dag}W)} U^{\dagger}W\log(U^{\dagger}W) e^{-(1-s)\log(U^{\dagger}W)}\,\rmd s\nonumber\\
	& = \log(U^{\dagger}W),
\end{align}
and consequently, using \eqref{eqn:expLogAdj},
\begin{align}
	\big(&\rmd_{U^{\dagger}W}\log\big)^{*}\big(\log(U^{\dagger}W)\big) \nonumber\\
	& = \big(\rmd_{U^{\dagger}W}\log\big)^{*}\circ \big(\rmd_{\log U^{\dagger}W}\exp\big)^{*}\big(U^{\dagger}W\log(U^{\dagger}W)\big)\nonumber\\
	& = U^{\dagger}W\log(U^{\dagger}W).
\end{align}
We therefore can rewrite the gradient of $J_{G}$ as 
\begin{align}
\grad J_{G}(U) & = -U\big[\big(\rmd_{U^{\dagger}W}\log\big)^{*}\big(\log(U^{\dagger}W)\big)\big]W^{\dagger}U\nonumber\\
& = -W\log(U^{\dagger}W)W^{\dagger}U\nonumber\\
& = -U\log(U^{\dagger}W).
\end{align}
where the last step follows from the fact that $W^{\dagger}U$ commutes with $\log(U^{\dag}W)$.

\bigskip

\bibliographystyle{IEEEtran}
\bibliography{criticalUnitary}

\begin{thebibliography}{10}
\providecommand{\url}[1]{#1}
\csname url@samestyle\endcsname
\providecommand{\newblock}{\relax}
\providecommand{\bibinfo}[2]{#2}
\providecommand{\BIBentrySTDinterwordspacing}{\spaceskip=0pt\relax}
\providecommand{\BIBentryALTinterwordstretchfactor}{4}
\providecommand{\BIBentryALTinterwordspacing}{\spaceskip=\fontdimen2\font plus
\BIBentryALTinterwordstretchfactor\fontdimen3\font minus
  \fontdimen4\font\relax}
\providecommand{\BIBforeignlanguage}[2]{{%
\expandafter\ifx\csname l@#1\endcsname\relax
\typeout{** WARNING: IEEEtran.bst: No hyphenation pattern has been}%
\typeout{** loaded for the language `#1'. Using the pattern for}%
\typeout{** the default language instead.}%
\else
\language=\csname l@#1\endcsname
\fi
#2}}
\providecommand{\BIBdecl}{\relax}
\BIBdecl

\bibitem{Rabitz2004}
\BIBentryALTinterwordspacing
H.~A. Rabitz, M.~M. Hsieh, and C.~M. Rosenthal, ``Quantum optimally controlled
  transition landscapes,'' \emph{Science}, vol. 303, pp. 1998--2001, Mar. 26
  2004. [Online]. Available:
  \url{http://www.sciencemag.org/content/303/5666/1998.abstract}
\BIBentrySTDinterwordspacing

\bibitem{Rabitz2005}
\BIBentryALTinterwordspacing
------, ``Landscape for optimal control of quantum-mechanical unitary
  transformations,'' \emph{Phys. Rev. A}, vol.~72, p. 052337, 2005. [Online].
  Available: \url{http://link.aps.org/doi/10.1103/PhysRevA.72.052337}
\BIBentrySTDinterwordspacing

\bibitem{Rabitz2006}
\BIBentryALTinterwordspacing
H.~A. Rabitz, T.-S. Ho, M.~M. Hsieh, R.~Kosut, and M.~Demiralp, ``Topology of
  optimally controlled quantum mechanical transition probability landscapes,''
  \emph{Phys. Rev. A}, vol.~74, p. 012721, 2006. [Online]. Available:
  \url{http://link.aps.org/doi/10.1103/PhysRevA.74.012721}
\BIBentrySTDinterwordspacing

\bibitem{Rabitz2006a}
\BIBentryALTinterwordspacing
H.~A. Rabitz, M.~M. Hsieh, and C.~M. Rosenthal, ``Optimal control landscapes
  for quantum observables,'' \emph{J. Chem. Phys.}, vol. 124, p. 204107, 2006.
  [Online]. Available: \url{http://link.aip.org/link/?JCP/124/204107/1}
\BIBentrySTDinterwordspacing

\bibitem{Wu2008}
\BIBentryALTinterwordspacing
R.~Wu, H.~Rabitz, and M.~Hsieh, ``Characterization of the critical submanifolds
  in quantum ensemble control landscapes,'' \emph{J. Phys. A: Math. Theor.},
  vol.~41, p. 015006, 2008. [Online]. Available:
  \url{http://stacks.iop.org/1751-8121/41/i=1/a=015006}
\BIBentrySTDinterwordspacing

\bibitem{Wu2008a}
\BIBentryALTinterwordspacing
R.~Wu, A.~Pechen, H.~Rabitz, M.~Hsieh, and B.~Tsou, ``Control landscapes for
  observable preparation with open quantum systems,'' \emph{J. Math. Phys.},
  vol.~49, p. 022108, 2008. [Online]. Available:
  \url{http://link.aip.org/link/?JMP/49/022108/1}
\BIBentrySTDinterwordspacing

\bibitem{Hsieh2008}
\BIBentryALTinterwordspacing
M.~Hsieh, R.~Wu, C.~Rosenthal, and H.~Rabitz, ``Topological and statistical
  properties of quantum control transition landscapes,'' \emph{J. Phys. B: At.
  Mol. Opt. Phys.}, vol.~41, no.~7, p. 074020, 2008. [Online]. Available:
  \url{http://stacks.iop.org/0953-4075/41/i=7/a=074020}
\BIBentrySTDinterwordspacing

\bibitem{Hsieh2008a}
\BIBentryALTinterwordspacing
M.~Hsieh and H.~Rabitz, ``Optimal control landscape for the generation of
  unitary transformations,'' \emph{Phys. Rev. A}, vol.~77, p. 042306, 2008.
  [Online]. Available: \url{http://link.aps.org/doi/10.1103/PhysRevA.77.042306}
\BIBentrySTDinterwordspacing

\bibitem{Ho2009}
\BIBentryALTinterwordspacing
T.-S. Ho, J.~Dominy, and H.~Rabitz, ``The landscape of unitary transformations
  in controlled quantum dynamics,'' \emph{Phys. Rev. A}, vol.~79, p. 013422,
  2009. [Online]. Available:
  \url{http://link.aps.org/doi/10.1103/PhysRevA.79.013422}
\BIBentrySTDinterwordspacing

\bibitem{Bott1954}
\BIBentryALTinterwordspacing
R.~Bott, ``Nondegenerate critical manifolds,'' \emph{Ann. of Math.}, vol.~60,
  no.~2, pp. 248--261, 1954. [Online]. Available:
  \url{http://www.jstor.org/stable/1969631}
\BIBentrySTDinterwordspacing

\bibitem{Atiyah1983}
\BIBentryALTinterwordspacing
M.~F. Atiyah and R.~Bott, ``The {Y}ang-{M}ills equations over {R}iemann
  surfaces,'' \emph{Philos. Trans. R. Soc. Lond. Ser. A Math. Phys. Eng. Sci.},
  vol. 308, no. 1505, pp. 523--615, 1983. [Online]. Available:
  \url{http://dx.doi.org/10.1098/rsta.1983.0017}
\BIBentrySTDinterwordspacing

\bibitem{Nicolaescu2007}
L.~Nicolaescu, \emph{An Invitation to {M}orse Theory}.\hskip 1em plus 0.5em
  minus 0.4em\relax New York: Springer, 2007.

\bibitem{Helmke1996}
U.~Helmke and J.~B. Moore, \emph{Optimization and Dynamical Systems}.\hskip 1em
  plus 0.5em minus 0.4em\relax Springer, 1996.

\bibitem{Beltrani2011}
\BIBentryALTinterwordspacing
V.~Beltrani, J.~Dominy, T.-S. Ho, and H.~Rabitz, ``Exploring the top and bottom
  of the quantum control landscape,'' \emph{J. Chem. Phys.}, vol. 134, no.~19,
  p. 194106, 2011. [Online]. Available:
  \url{http://link.aip.org/link/doi/10.1063/1.3589404}
\BIBentrySTDinterwordspacing

\bibitem{Wu2012}
\BIBentryALTinterwordspacing
R.~Wu, R.~Long, J.~Dominy, T.-S. Ho, and H.~Rabitz, ``Singularities of quantum
  control landscapes,'' \emph{Phys. Rev. A}, vol.~86, no.~1, p. 013405, 2012.
  [Online]. Available: \url{http://link.aps.org/doi/10.1103/PhysRevA.86.013405}
\BIBentrySTDinterwordspacing

\bibitem{Fouquieres}
\BIBentryALTinterwordspacing
P.~de~Fouquieres and S.~G. Schirmer, ``Quantum control landscapes: A closer
  look,'' 2010. [Online]. Available: \url{http://arxiv.org/abs/1004.3492}
\BIBentrySTDinterwordspacing

\bibitem{Pechen2011}
\BIBentryALTinterwordspacing
A.~N. Pechen and D.~J. Tannor, ``Are there traps in quantum control
  landscapes?'' \emph{Phys. Rev. Lett.}, vol. 106, p. 120402, Mar. 2011.
  [Online]. Available:
  \url{http://link.aps.org/doi/10.1103/PhysRevLett.106.120402}
\BIBentrySTDinterwordspacing

\bibitem{Rabitz2012}
\BIBentryALTinterwordspacing
H.~Rabitz, T.-S. Ho, R.~Long, R.~Wu, and C.~Brif, ``Comment on ``are there
  traps in quantum control landscapes?'','' \emph{Phys. Rev. Lett.}, vol. 108,
  p. 198901, May 2012. [Online]. Available:
  \url{http://link.aps.org/doi/10.1103/PhysRevLett.108.198901}
\BIBentrySTDinterwordspacing

\bibitem{Palao2003}
\BIBentryALTinterwordspacing
J.~P. Palao and R.~Kosloff, ``Optimal control theory for unitary
  transformations,'' \emph{Phys. Rev. A}, vol.~67, no.~6, p. 062308, Dec. 2003.
  [Online]. Available: \url{http://link.aps.org/doi/10.1103/PhysRevA.68.062308}
\BIBentrySTDinterwordspacing

\bibitem{doCarmo1992}
M.~P. do~Carmo, \emph{Riemannian Geometry}.\hskip 1em plus 0.5em minus
  0.4em\relax Boston: Birkh{\"a}user, 1992.

\bibitem{Warner1983}
F.~W. Warner, \emph{Foundations of Differentiable Manifolds and {L}ie
  Groups}.\hskip 1em plus 0.5em minus 0.4em\relax New York: Springer, 1983.

\bibitem{Knapp2004}
A.~W. Knapp, \emph{Lie Groups: Beyond an Introduction}, 2nd~ed.\hskip 1em plus
  0.5em minus 0.4em\relax Boston: Birkh{\"a}user, 2004.

\bibitem{Helgason2001}
S.~Helgason, \emph{Differential Geometry, Lie Groups, and Symmetric
  Spaces}.\hskip 1em plus 0.5em minus 0.4em\relax Providence, RI: American
  Mathematical Society, 2001.

\bibitem{Milnor1973}
J.~Milnor, \emph{Morse Theory}, ser. Annals of Mathematics Studies.\hskip 1em
  plus 0.5em minus 0.4em\relax Princeton: Princeton University Press, 1973,
  vol.~51.

\bibitem{Abraham1988}
R.~Abraham, J.~E. Marsden, and T.~Ratiu, \emph{Manifolds, Tensor Analysis, and
  Applications}, 2nd~ed.\hskip 1em plus 0.5em minus 0.4em\relax New York:
  Springer, 1988.

\bibitem{Horn1985}
R.~A. Horn and C.~R. Johnson, \emph{Matrix Analysis}.\hskip 1em plus 0.5em
  minus 0.4em\relax New York: Cambridge University Press, 1985.

\bibitem{Dominy2008}
\BIBentryALTinterwordspacing
J.~Dominy and H.~Rabitz, ``Exploring families of controls for generating
  unitary transformations,'' \emph{J. Phys. A: Math. Theor.}, vol.~41, no.~20,
  p. 205305, May 2008. [Online]. Available:
  \url{http://stacks.iop.org/1751-8121/41/i=20/a=205305}
\BIBentrySTDinterwordspacing

\bibitem{Abramowitz1972}
M.~Abramowitz and I.~Stegun, \emph{Handbook of Mathematical Functions}.\hskip
  1em plus 0.5em minus 0.4em\relax New York: Dover, 1972.

\bibitem{Karplus1948}
\BIBentryALTinterwordspacing
R.~Karplus and J.~Schwinger, ``A note on saturation in microwave
  spectroscopy,'' \emph{Phys. Rev.}, vol.~73, pp. 1020--1026, May 1948.
  [Online]. Available: \url{http://link.aps.org/doi/10.1103/PhysRev.73.1020}
\BIBentrySTDinterwordspacing

\bibitem{Mathias1992}
\BIBentryALTinterwordspacing
R.~Mathias, ``Evaluating the {F}rechet derivative of the matrix exponential,''
  \emph{Numer. Math.}, vol.~63, pp. 213--226, 1992. [Online]. Available:
  \url{http://dx.doi.org/10.1007/BF01385857}
\BIBentrySTDinterwordspacing

\end{thebibliography}
	
\end{document}